\documentclass[a4paper,reqno]{amsart}

\usepackage{a4wide}
\usepackage{verbatim}

\newcommand{\J}{\mathbb{I}}

\newcommand{\nml}{\mathbf{n}}
\newcommand{\sphere}{\mathbb{S}}
\newcommand{\setR}{\mathbb{R}}
\newcommand{\setN}{\mathbb{N}}

\newcommand{\E}{\mathbb{E}}
\newcommand{\eps}{\epsilon}
\newcommand{\eins}{\mathbf{1}}
\newcommand{\tr}{\operatorname{tr}}

\newcommand{\law}{\operatorname{Law}}

\newcommand{\tq}{\Upsilon}
\newcommand{\ee}{\mathbf{e}}
\newcommand{\cov}{\operatorname{Cov}}
\newcommand{\cls}{\mathcal{C}}
\newcommand{\CB}{\mathcal{B}}
\newcommand{\CF}{\mathcal{F}}

\newcommand{\zz}{\mathbf{z}}
\newcommand{\smat}{\mathcal{S}}
\newcommand{\meas}{\mathcal{P}}
\newcommand{\fnorm}{\mathbf{D}}
\newcommand{\newd}{\mathbf{d}}
\newcommand{\so}{\operatorname{SO}}
\newcommand{\half}{{\textstyle{\frac12}}}

\newcommand{\prb}{\mathbb{P}}
\newcommand{\sigalg}{\mathcal{B}}
\newcommand{\indi}{\mathbf{1}}
\newcommand{\tmp}{\mathbf{\Theta}}
\newcommand{\distreq}{\stackrel{\mathcal{D}}{=}}

\newcommand{\ball}{\mathbf{B}}
\newcommand{\indy}{\mathbf{1}}
\newcommand{\xx}{X}
\newcommand{\xu}{\underline{X}}
\newcommand{\xo}{\overline{X}}

\newcommand{\LL}{l}
\newcommand{\RR}{r}

\newcommand{\weight}{\omega} 
\newcommand{\wweight}{\mathcal{W}} 

\newcommand{\sgn}{\operatorname{sgn}}
\newcommand{\dd}{\mathrm{d}}
\newcommand{\Si} {\Sigma}

\newtheorem{thm}{Theorem}
\newtheorem{lem}{Lemma}
\newtheorem{prp}{Proposition}

\newtheorem{assum}{Assumption}
\newtheorem{rmk}{Remark}

\title[Homogeneous kinetic equations]{Homogeneous kinetic equations 
  for probabilistic linear collisions in multiple space dimensions}

\author{Federico Bassetti}
\address{F.~Bassetti \\ Dipartimento di Matematica \\ Universit\`a di Pavia \\ Via Ferrata 1 \\ 27100 Pavia \\ Italy}
\email{federico.bassetti@unipv.it}

\author{Daniel Matthes}
\address{D.~Matthes \\ Zentrum Mathematik \\ Technische Universit\"at M\"unchen \\ Boltzmannstra\ss e 3 \\ 85748 Garching bei M\"unchen \\ Germany}
\email{matthes@ma.tum.de}

\begin{document}

\begin{abstract}
  We analyze the convergence to equilibrium in a family of Kac-like kinetic equations in multiple space dimensions.
  These equations describe the change of the velocity distribution in a spatially homogeneous gas 
  due to binary collisions between the particles.
  We consider a general linear mechanism for the exchange of the particles' momenta,
  with interaction coefficients that are random matrices with a distribution that is {independent} of the velocities of the colliding particles.
 Applying a synthesis of probabilistic methods and Fourier analysis,
  we are able to identify sufficient conditions for the existence and uniqueness of a stationary state,
  we characterize this stationary state as a mixture of Gaussian distributions,
  and we prove equilibration of transient solutions under minimal hypotheses on the initial conditions.
  In particular, we are able to classify the high-energy tails of the stationary distribution,
  which might be of Pareto type.
  We also discuss several examples to which our theory applies,
  among them models with a non-symmetric stationary state.
\end{abstract}

\maketitle

\section{Introduction}

This paper is concerned with the spatially homogeneous Boltzmann equation
\begin{align}
  \label{eq.b}
  \partial_t\mu + \mu = Q_+[\mu], \qquad \mu(0)=\mu_0,
\end{align}
for a time-dependent probability distribution $\mu(t,\dd v)$ on the particle velocities $v\in\setR^d$ in a gas. 
Above, the collisional gain operator $Q_+[\mu]\equiv Q_+[\mu,\mu]$ is bi-linear and given (in weak form) by
\begin{align}
  \label{eq.gain}
  \int_{\setR^d} \psi(v')  Q_+[\mu_1,\mu_2](\dd v')
  = \frac12 \int_{\setR^d\times\setR^d} \E\big[\psi(v')+\psi(v_*') \big]\, \mu_1(\dd v)\, \mu_2(\dd v_*) ,
\end{align}
for every $\psi\in C_b(\setR^d)$. 
It accounts for the change in velocity due to binary collisions between particles.
Taking the expectation $\E$ in \eqref{eq.gain} corresponds --- morally --- to integration with respect to the scattering angle.
Our main assumptions are the following:
\begin{enumerate}
\item Collisions occur at {unit frequency}, i.e., the probability for a particle to collide is independent of its velocity.
  This justifies the particular form \eqref{eq.b} of the Boltzmann equation.
\item The post-collisional velocities $v'$ and $v_*'$ are determined from the pre-collisional velocities $v$ and $v_*$ by {linear} rules
  \begin{align}
    \label{eq.rules}
    v'=Lv+Rv_*, \quad v_*' = L_*v_*+R_*v,
  \end{align}
  in which $L$, $R$, $L_*$ and $R_*$ are $d\times d$ random matrices of a \emph{given distribution},
  that is independent of $v$ and $v_*$.
  For simplicity, we assume that $(L,R)$ and $(L_*,R_*)$ are identically distributed,
  but we do not assume their stochastic independence.
\item The model conserves the second absolute moment of $\mu(t;\dd v)$, 
  which means that
  \begin{align}
    \label{eq.keepenergy}
    \E\big[ |v'|^2 + |v_*'|^2 \big] = |v|^2 + |v_*|^2 .
  \end{align}
  Thus, the total kinetic energy (or temperature) is a constant.
\end{enumerate}
Equation \eqref{eq.b} with kernel \eqref{eq.gain} constitutes one possible generalization
of the celebrated Kac model to multiple space dimensions.
The original Kac model
--- which is \eqref{eq.b} and \eqref{eq.gain} in $d=1$ dimension with $(L,R)=(\cos\phi,\sin\phi)$,
and $\phi$ uniformly distributed on $[0,2\pi)$ ---
has been introduced \cite{Kac} as a caricature of the Boltzmann equation
to investigate the propagation of chaos in particle systems.
In multiple space dimensions,
\emph{Maxwell molecules} \cite{Bobylev} are the commonly accepted canonical ``interpolation''
between the simplified Kac model and the genuine Boltzmann equation.
Also in the Maxwell case, one uses 
linear rules of the type \eqref{eq.rules}, which is the so-called $\omega$-representation of collisions,
and a uniform collision frequency, just as we do here.
The difference is that for Maxwell molecules, 
the probability distribution of the coefficients $(L,R)$ typically depends on the spatial direction of $v_*-v$.
(Maxwell molecules \emph{without} this dependence are a special case of the model considered here,
see Section \ref{sct.maxwell}.)

Our approach is not an attempt to bridge the gap between the Kac model and the physical Boltzmann equation in an alternative way.
Rather, we introduce a novel class of binary interaction models, which are reminiscent of the Kac equation,
but with distinguished features that makes this class interesting on its own right.
We emphasize that there is a huge freedom in the choice of the matrix coefficients $(L,R)$.
The theory that we develop does not require the conservation of momentum or energy in binary collisions,
but only the much weaker stochastic condition \eqref{eq.keepenergy}.
Neither do we ask for any covariance under rotations of $\setR^d$.
Binary particle collisions without momentum/energy conservation have been used before,
for instance to combine particle collisions and the influence of a heat bath into a single interaction mechanism \cite{CorCarTos},
as well as in applications of kinetic theory to simple models for wealth distribution \cite{DurMatTos,ParTos}.
The violation of rotational symmetry, on the other hand, can be thought of as resulting from a non-isotropic background,
which induces an a priori preference for certain post-collisional directions on the interacting particles.
As a consequence of its high degree of freedom, 
there is a huge variety of potential stationary states for the velocity distributions,
the details of which depend sensitively on the model parameter.
These stationary distributions might possess high energy tails of Pareto type,
and are not necessarily rotationally symmetric.

For the development of the theory, we continue in the spirit of various recent works,
like \cite{BasLadReg,BobCerGam,GabettaRegazziniCLT,GabettaRegazziniWM},
which have been concerned with the applicability of probabilistic methods 
to study the equilibration in kinetic systems.
Specifically, our approach parallels in large parts the one taken recently in \cite{BasLad,BasLadMat} 
for the one-dimensional reduction of \eqref{eq.b}.
In one spatial dimension, when $L$ and $R$ are simply non-negative random numbers $\ell$ and $r$,
it has been proven that the conditions
\begin{align}
  \label{eq.1d1}
  \E[\ell^2+r^2] &= 1, \quad \text{and}\\
  \label{eq.1d2}
  \E[\ell^p+r^p] &< 1 \quad \text{for some $p>2$},
\end{align}
are sufficient for existence and uniqueness of a stationary distribution $\mu_\infty$ of finite second moment.
Moreover, if \eqref{eq.1d1} and \eqref{eq.1d2} are satisfied, 
then finiteness of the second moment of the initial condition $\mu_0$ is sufficient 
for (weak) convergence of the transient solution $\mu(t)$ to $\mu_\infty$ as $t\to\infty$.

Here, we are going to prove an analogous statement in the multi-dimensional situation $d>1$.
However, the substitutes for conditions \eqref{eq.1d1} and \eqref{eq.1d2} are not obvious.
In fact, the most straight-forward generalizations of these conditions to matrices $L$ and $R$ would not work.
Moreover, our proof needs a third condition, see \eqref{eq.ass2} below.
We are now going to state and briefly discuss the individual conditions.

The first condition, which parallels \eqref{eq.1d1}, accounts for our hypothesis \eqref{eq.keepenergy}  
that the second absolute moment is constant in time.
\begin{assum}
  \label{ass.1}
  The matrix coefficients $L$ and $R$ satisfy
  \begin{align}
    \label{eq.general}
    \E[L^TL+R^TR] = \eins ,
  \end{align}
  where $\eins\in\setR^{d\times d}$ denotes the unit matrix. 
\end{assum}
It is tempting to substitute \eqref{eq.1d2} by the condition
\begin{align*}
  \E[\|L\|^p+\|R\|^p]< 1
\end{align*}
with, say, the operator norm $\|A\|=\sup_{|\ee|=1} |A\ee|$.
In general, however, this condition would be too strict.
It is not even satisfied in the paradigmatic example of Maxwell molecules 
discussed in Section \ref{sct.maxwell}.
Instead, the following appears to be the most appropriate substitute.
\begin{assum}
  \label{ass.3new}
  There exist real numbers $p>2$, $\bar\weight\ge1$ and $\kappa_p\in(0,1)$,
  and a \emph{weight function} $\weight:\setR^d\to\setR$ with the following properties:
  for every $\eta\in\setR^d$,
  \begin{align}
    \label{eq.phomogen}
    & \weight(\lambda\eta)=\lambda^p\weight(\eta) \quad \text{for all $\lambda\geq0$}, \\
    \label{eq.boundass}
    & |\eta|^p \leq \weight(\eta) \leq \bar \weight |\eta|^p  , \quad \text{and} \\
    \label{eq.weirdass}
    & \E[\weight(L^T\eta) + \weight(R^T\eta)] \leq \kappa_p \,\weight(\eta).
  \end{align}
  We refer to \eqref{eq.phomogen} and \eqref{eq.boundass}
  as \emph{$p$-homogeneity} and \emph{boundedness} of $\weight$, respectively.
\end{assum}
This condition is employed at various points of the discussion;
in combination with Assumption \ref{ass.2} below,
it provides contractivity estimates in suitable metrics and weak compactness of transient solutions.
Notice that $\weight(\xi)=|\xi|^p$ defines a $p$-homogeneous and bounded weight function.
In Section \ref{sct.cross} we give a multi-dimensional example
in which a more complicated weight function is needed in order to satisfy \eqref{eq.weirdass}.

Our last hypothesis has no analogue in one dimension.
It guarantees that the multi-dimensional model cannot be decomposed into a family of one-dimensional sub-systems. 
\begin{assum}
  \label{ass.2}
  The map $\tq$, defined on symmetric matrices $M\in\setR^{d\times d}$ by
  \begin{align}
    \tq(M) = \E[LML^T+RMR^T],
  \end{align}
  possesses a fixed point $\Sigma_*=\tq(\Sigma_*)$
  that is positive definite and has $\tr\Sigma_*=d$.
  Moreover, there is a constant $\kappa<1$ 
  such that
  \begin{align}
    \label{eq.ass2}
    \| \tq(M) \|_\infty \leq \kappa \|M\|_\infty .
  \end{align}
  holds for every \emph{traceless} symmetric $d\times d$-matrix $M$.
\end{assum}
Here $\|M\|_\infty$ denotes the spectral radius of the matrix $M$;
see Section \ref{sct.norms} for its definition.
An immediate implication of \eqref{eq.ass2} --- in combination with \eqref{eq.general} --- 
is that the fixed points of $\tq$ are scalar multiples of $\Sigma_*$.
This condition is needed to ensure that all ``directional information'' of the initial condition 
is lost in the long-time limit.

Assumption \ref{ass.2} --- although not very intuitive --- appears to be essential
for the development of the theory.
It accounts for the fact that we are \emph{not} working in the one-dimensional or radially symmetric situation,
like in \cite{BasLadMat} or \cite{BobCerGam}, respectively.
Indeed, one could think of a straight-forward generalization 
of the one-dimensional/radially symmetric concepts from \cite{BasLadMat} or \cite{BobCerGam} 
to general solutions in $\setR^d$ by using
\begin{align*}
  L=\ell\eins, \quad R=r\eins, 
\end{align*}
where $\ell$ and $r$ are non-negative random variables satisfying \eqref{eq.1d1}\&\eqref{eq.1d2}.
Clearly, while Assumptions \ref{ass.1} and \ref{ass.3new} hold in this case,
Assumption \ref{ass.2} is violated,
since
\begin{align*}
  \| \E[LML^T+RMR^T] \|_\infty = \| \E[\ell^2+r^2]M \|_\infty = \|M\|_\infty
\end{align*}
for every matrix $M$.
By the results of \cite{BasLadMat}, one can associate
to every one-dimensional linear subspace $\operatorname{span}(w)\subset\setR^d$ with $0\neq w\in\setR^d$
a non-trivial stationary distribution $\mu_\infty^w$,
which attracts all transient solutions $\mu(t)$ to \eqref{eq.b} 
that are initially supported on $\operatorname{span}(w)$ and have unit second moment.
Consequently, there are infinitely many stationary distributions.

We shall now summarize our main results.
Below, we mean by \emph{temperature} of a centered probability distribution $\mu$ on $\setR^d$
the quantity
\begin{align*}
  \tmp[\mu] = \frac1d \int_{\setR^d} |v|^2\mu(\dd v).
\end{align*}
The normalization is chosen such that
the Gaussian distribution with unit covariance matrix is of unit temperature.
%

\begin{thm}
  \label{thm.main}
  Under Assumptions 1--3,
  equation \eqref{eq.b} possesses precisely one stationary distribution $\mu_\infty=Q_+[\mu_\infty]$
  that is centered and of unit temperature.
  This distribution $\mu_\infty$ is a scale mixture of Gaussians, i.e.,
  \begin{align}
    \label{eq.fmix}
    \int_{\setR^d} \psi(v) \mu_\infty(\dd v)
    =\E \left [\int_{\setR^d} \psi\big(\sqrt{S}w) \frac{e^{-\frac{1}{2}|w|^2 }}{(2\pi)^{d/2}} \dd w \right ]
  \end{align}
  for every test function $\psi \in C_b(\setR^d)$, 
  where $S$ is a symmetric and positive semi-definite random matrix in $\setR^{d\times d}$,
  and $\sqrt{S}$ denotes its positive semi-definite square root. 
  Moreover, the $p$th absolute moment of $\mu_\infty$ is finite.

  If $\mu_0$ is a centered initial distribution with finite temperature,
  then the corresponding transient solution $\mu(t)$ of the initial value problem \eqref{eq.b} converges weakly
  to a non-trivial stationary distribution $\mu_*$,  and $\mu_*$ is obtained from $\mu_\infty$ 
  by the unique isotropic rescaling for which $\tmp[\mu_*]=\tmp[\mu_0]$.

  If in addition $\prb\{L^T\ee=R^T\ee=0\}=0$ and $\prb\{L^T\ee\neq0\neq R^T\ee\}>0$ for every unit vector $\ee$,
  then $\mu_\infty$ is absolutely continuous.
  In this case, if $\mu_0$ is a centered initial distribution then the corresponding transient solution $\mu(t)$ converges weakly
  to $\mu_*$ if and only if $\mu_0$ has finite temperature.
\end{thm}
Apart from Theorem \ref{thm.main},
we prove results on the Sobolev-regularity of $\mu_\infty$'s density (Theorem \ref{thm.smooth}),
on the rate of convergence of $\mu(t)$ to $\mu_\infty$ in Fourier metrics (Theorem \ref{thm.fourierconverge}),
and on the divergence of higher moments (Proposition \ref{prp.moments}).

Finally, we outline the structure of the paper and give an overview on the applied techniques.

Section \ref{sct.prelim} collects various definitions and preliminary observations.
In particular, we introduce a family of matrix norms. 

Section \ref{sct.stationary} is concerned with the existence and certain properties 
of the stationary distribution $\mu_\infty$.
Assuming that $\mu_\infty$ is of the form \eqref{eq.fmix},
we first show that the random matrix $S$ needs to satisfy the fixed point equation
\begin{align}
  \label{eq.smoothing}
 \law(S) = \law(LS'L^T + RS''R^T) .
\end{align}
The right-hand side of \eqref{eq.smoothing} 
--- where $S'$ and $S''$ are i.i.d.\ copies of $S$, which are independent of $(L,R)$ ---
is a multi-dimensional generalization of a smoothing transformation.
For (scalar) random variables, 
fixed points of such transformations have been intensively studied in the literature,
see e.g., \cite{DurrettLiggett1983,iksanov,Liu1998,Liu2000}.
Some multidimensional generalizations of the smoothing transformation  have been treated  in 
\cite{Barral,Neininger}, although there are apparently no results on equations of the form \eqref{eq.smoothing}.

In Theorem \ref{thm.exuniq}, we obtain the (essentially unique) solvability of \eqref{eq.smoothing}
under Assumptions 1--3.
The proof uses a generalization of \emph{Fourier metrics}.
Originally, these define a distance between probability measures on $\setR^d$
from a suitably weighted difference of their respective Fourier transforms;
see \cite{CarTos} for a survey.
The basic ideas have been developed in \cite{Bobylev} and \cite{GabTosWen},
and since then applied --- with appropriate modifications --- in various context
related to homogeneous kinetic equation of Kac type, see e.g. \cite{BobCerGam,CarGabTos,DurMatTos}.
Our generalization to probability measures on matrices uses the same basic concept,
but is technically more involved.

Subsequently, results on symmetries, moments and further properties of $\mu_\infty$ are derived.
The most interesting auxiliary result is probably that on the existence and regularity of a density.
Since $\mu_\infty$ is not single Gaussian in general, but a mixture of the form \eqref{eq.smix},
it might possess concentrations.
In one spatial dimension, one can generally decompose any such scale mixture
into a singular measure concentrated at the origin,
and a regular measure with a density that is $C^\infty$ on $\setR\setminus\{0\}$,
but in general only $L^1$ on $\setR$.
We refer to \cite{MatTos} for a discussion on the Sobolev regularity of stationary states,
and on the propagation of smoothness in \eqref{eq.b} with \eqref{eq.gain},
and to \cite{CarGabTos,DesFurTer,FPTT} for closely related results on Maxwellian molecules. 

The situation is more difficult in multiple dimensions, since
concentrations can occur not only at the origin,
but on any lower-dimensional linear subspace, or on the union of subspaces.
A non-trivial example is given in Section \ref{sct.cross}.
We do not discuss the variants of possible concentrations in detail, 
but instead provide a sufficient condition for absolute continuity of the stationary distribution in Theorem \ref{thm.smooth}.
Its proof is a multi-dimensional extension
of the particularly elegant approach made in \cite{Liu}.

In Section \ref{sct.dynamic}, 
we prove that arbitrary solutions $\mu(t)$ to the initial value problem \eqref{eq.b} converge to $\mu_\infty$
if they are initially of unit temperature.
We use two approaches that lead to slightly different results.

The first result, see Theorem \ref{thm.clt}, is obtained by probabilistic methods.
The key element is a stochastic representation of $\mu(t)$ as the law of a weighted sum of independent random variables.
The long-time asymptotics of $\mu(t)$ can then be obtained
by studying the convergence of these weighted sums.
The latter is conveniently done in the framework of the central limit theorem.
The idea goes back essentially to McKean \cite{McKean1966,McKean1967}, 
and has been brought to full power in the recent works \cite{DolGabReg,dolera2}.
We are able to adapt large parts of the general strategy 
--- that has been successfully employed in a related one-dimensional situation \cite{BasLadMat} ---
to our current needs.
The result is \emph{qualitative} as it gives no direct information on the rate of convergence of $\mu(t)$ to $\mu_\infty$;
the advantage of this proof is that it works under the minimal hypothesis on the initial condition $\mu_0$, 
namely finiteness of the second moment.

The second result on convergence in Theorem \ref{thm.fourierconverge}
makes the previous statement \emph{quantitative} as it provides an exponential rate for 
the convergence of $\mu(t)$ to $\mu_\infty$ in Fourier distance.
The additional requirement is finiteness of the $p'$th absolute moment of $\mu_0$ for some $p'>2$,
and the exponential rate depends sensitively on $p'$.
The proof, which closely follows a by now classical strategy, is only sketched.

The proofs are carried out for an initial datum of unit temperature, $\tmp[\mu_0]=1$,
in which case the associated transient solution converges weakly towards the stationary distribution $\mu_\infty$
with covariance matrix $\cov\mu_\infty=\Sigma_*$, the fixed point of $\tq$.
A simply scaling argument shows that 
if $\vartheta:=\tmp[\mu_0]$ is not one but finite, 
then the transient solution still converges to a stationary distribution $\mu_\infty^\vartheta$,
which is simply a rescaling of $\mu_\infty$, with covariance matrix $\cov\mu_\infty^\vartheta=\vartheta\Sigma_*$.

In Section \ref{sct.explosion}, we analyze the limiting case $\tmp[\mu_0]=+\infty$.
The transient solution $\mu(t)$ exists for all times, but it does not converge weakly to a limit measure as $t\to\infty$.
Loosely speaking, $\mu(t)$ concentrates at infinity instead.
This statement has been made precise in the case of the spatially homogeneous Boltzmann equation 
for pseudo Maxwellian molecules in \cite{CaGaReExplosion} and for the Kac equation in \cite{explosion1}:
the restriction of $\mu(t)$ to any bounded set in phase space vanishes in the long time limit.

In Theorem \ref{thm.explosion}, we prove --- based on the ideas from \cite{CaGaReExplosion,explosion1} ---
that the situation is the same for the class of equations considered here.
The main idea is to consider a restriction $\bar\mu_0$ of the initial condition $\mu_0$ to a large ball,
such that $\bar\mu_0$ has finite but huge temperature $\vartheta = \tmp[\bar\mu_l] \gg 1$,
and compare the original transient solution $\mu(t)$ 
to the one associated to the initial value $\bar\mu_0$.
Combining the probabilistic representation with the quantitative convergence results from Section \ref{sct.dynamic},
we thus obtain an upper bound on the density of $\mu(t)$ --- for sufficiently large $t$ ---
by a multiple of the rescaling $\mu_\infty^\vartheta$ of $\mu_\infty$.
For $\vartheta\to\infty$, this upper bound vanishes,
and consequently, the density of $\mu(t)$ must converge locally uniformly to zero as $t\to\infty$.

Finally, Section \ref{sct.examples} provides several examples to which the developed theory applies.

\section{Notations and preliminary results}
\label{sct.prelim}

\subsection{Conventions}
For simplicity, we shall assume throughout the paper 
that the random coefficients $L$ and $R$ are symmetric in the sense
that $(L,R)\distreq(R,L)$.
Above and in the rest of the paper, $Z\distreq Z'$ means
that the random elements $Z$ and $Z'$ have the same distribution, i.e. $\law(Z)=\law(Z')$.

\subsection{Spaces of probability measures}\label{sct.norms}
By $\meas(X)$, we denote the space of probability measures on $X$.
Introduce the class $\cls_1(\setR^d)$ of centered probability measures on $\setR^d$,
\begin{align*}
  \cls_1(\setR^d):=\bigg\{ \mu\in\meas(\setR^d)\bigg|\int_{\setR^d} v\,\mu(\dd v)=0  \bigg\},
\end{align*}
and for $q\geq2$ the subclasses
\begin{align*}
  \cls_q(\setR^d):=\bigg\{ \mu\in\cls_1(\setR^d)\,\bigg|
  \int_{\setR^d}|v|^2\, \mu(\dd v)=d,\,\int_{\setR^d} |v|^q\, \mu(\dd v)<\infty \bigg\}
\end{align*}
of measures with unit temperature and finite $q$th absolute moment.
The natural topology on $\cls_q(\setR^d)$ is the one induced by weak convergence of probability measures,
plus convergence of the $q$th absolute moment.

As usual, we denote by the $\setR^{d\times d}$ the set of all real $d\times d$-matrices.
Introduce the subsets of symmetric and positive semi-definite matrices by
\begin{align*}
  \smat = \big\{\Si \in\setR^{d\times d} \big| \Si^T=\Si \big\}, 
  \quad
  \smat_+= \big\{\Si \in\smat \big| v^T \Si v \geq 0\, \text{ for all $v\in\setR^d$} \big\},
\end{align*}
and define accordingly for $p\ge2$
\begin{align*}
  \cls_p(\smat_+) = \left \{ \nu\in\meas(\smat_+) \bigg|  \int_{\smat_+} 
    (\tr\Si)\,\nu( \dd \Si) =d,\, \int_{\smat_+} (\tr \Si)^{\frac{p}{2}} \,\, \nu (\dd \Si)<\infty \right \} ,
\end{align*}
the probability measures on $\smat_+$ with finite $p$th absolute moment and normalized first moment.

\subsection{Matrix norms}
Several different norms for matrices will be used in the sequel.
First recall that a real symmetric matrix $\Si \in\smat$ always possesses a spectral decomposition of the form
\begin{align*}
  \Si = \sum_{j=1}^d \lambda_j\ee_j\ee_j^T,
\end{align*}
where $\{\ee_j\}_{j=1}^d$ is an orthonormal basis of $\setR^d$ of eigenvectors of $\Si$,
and the $\lambda_j$ are the corresponding (real) eigenvalues.
For $q\ge1$, introduce the $q$th matrix norm by
\begin{align*}
  \|\Si\|_q := \Big( \sum_{j=1}^d |\lambda_j|^q \Big)^{1/q},
\end{align*}
as well as
\begin{align*}
  \|\Si\|_\infty := \max_{i=1,...,d} |\lambda_i| = \sup_{|\ee|=1} |\ee^T\Si\ee|.
\end{align*}
All norms $\|\cdot\|_q$ are equivalent.
In particular, the 2-norm is induced by the scalar product
\begin{align*}
  \tr(\Xi^T \Si) = \sum_{i,j=1}^d \Xi_{ij}\Si_{ij}.
\end{align*}
The following submultiplicativity estimates will be used,
\begin{align}
  \label{eq.matrixholder}
  |\tr(\Xi^T\Si)| &\le \|\Xi\|_1 \|\Si\|_\infty, \\
  \label{eq.submult}
  |\tr(A^T\Si A)| &\leq \|\Si\|_2\|A\|_2^2 , \\
  \label{eq.submult2}
  |v^T\Si w| &\le \|\Si\|_\infty |v| |w|,
\end{align}
which hold for arbitrary $\Si,\Xi\in\smat$, $A\in\setR^{d\times d}$ and $v,w\in\setR^d$.

\subsection{Weight functions}
In several proofs, we shall work under the simplifying hypothesis 
that Assumption \ref{ass.3new} holds with $p\in(2,3)$.
We shall now show that there is no loss of generality in doing so.
\begin{lem}
  \label{lem.prime}
  For every $p'\in(2,p)$, Assumption \ref{ass.3new} is satisfied with the modified weight function
  \begin{align*}
    \weight'(\xi) = \zeta_-^{-\epsilon}\weight(\xi)^{1-\epsilon}(\xi^T\Sigma_*\xi)^\epsilon, \quad 
    \text{where} \quad \epsilon:=\frac{p-p'}{p-2} \in(0,1),
  \end{align*}  
  and constants $\bar\omega'=(\zeta_+/\zeta_-)^\epsilon\bar\omega^{1-\epsilon}\ge1$, 
  $\kappa_{p'}=\kappa_p^{1-\epsilon}\in(0,1)$,
  where $\zeta_+,\zeta_-\in\setR_+$ are the biggest and smallest eigenvalue of $\Sigma_*$, respectively.
\end{lem}
\begin{proof}
  One has $\weight'(\lambda\xi)=\lambda^{p(1-\epsilon)}\lambda^{2\epsilon}\weight'(\xi)=\lambda^{p'}\weight'(\xi)$ 
  for all $\xi\in\setR^d$ and $\lambda>0$, showing \eqref{eq.phomogen}.
  And from $\zeta_-|\xi|^2 \le \xi^T\Sigma_*\xi \le \zeta_+|\xi|^2$ it follows that
  $(\zeta_-/\zeta_-)^\epsilon|\xi|^{p'}\leq\weight'(\xi)\leq(\zeta_+/\zeta_-)^\epsilon\bar\weight^{1-\epsilon}\zeta_+^\epsilon|\xi|^{p'}$,
  which is \eqref{eq.boundass}.
  Finally, by H\"older's inequality and the symmetry $(L,R)\distreq(R,L)$,
  \begin{align*}
    \E[\weight'(L^T\xi)+\weight'(R^T\xi)]
    = 2\E[\weight'(L^T\xi)] 
    & \leq \zeta_-^{-\epsilon}\big(2\E[\weight(L^T\xi)]\big)^{1-\epsilon}\big(2\E[\xi^TL\Sigma_*L^T\xi]\big)^\epsilon \\
    & = \zeta_-^{-\epsilon}\E[\weight(L^T\xi)+\weight(R^T\xi)]^{1-\epsilon}\big(\xi^T\E[L\Sigma_*L^T+R\Sigma_*R^T]\xi\big)^\epsilon \\
    & \leq \zeta_-^{-\epsilon}\kappa_p^{1-\epsilon}\weight(\xi)(\xi^T\Sigma_*\xi)^\epsilon = \kappa_p^{1-\epsilon}\weight'(\xi).
  \end{align*}
  This confirms \eqref{eq.weirdass}.
\end{proof}

From the weight function $\weight$ on $\setR$, 
define the following weight function on symmetric matrices $\Xi\in\smat$:
\begin{align}
  \label{eq.defww}
  \wweight(\Xi) = \inf \Big\{ \sum_{j=1}^d\weight(\xi_j) \,\Big|\, 
  \Xi = \sum_{j=1}^d \tau_j\xi_j\xi_j^T \text{ with }\xi_1,\ldots,\xi_d\in\setR^d,\,\tau_1,\ldots,\tau_d\in\{\pm1\} \Big\}.
\end{align}
The next lemma lists useful properties of $\wweight$;
recall the definition of $p>2$ from Assumption \ref{ass.3new}.
\begin{lem}
 The weight function $\wweight$
  \begin{itemize}
  \item[(i)] is homogeneous of degree $p/2$:
    \begin{align}
      \label{eq.wwhom}
      \wweight(\lambda\Xi) = \lambda^{p/2}\wweight(\Xi) \quad \text{for all $\Xi\in\smat,\,\lambda>0$};
    \end{align}
  \item[(ii)] is compatible with the usual matrix norms on the set $\smat$:
    \begin{align}
      \label{eq.wwcompatible}
      d^{-(p/2-1)} \|\Xi\|_1^{p/2} \leq \wweight(\Xi) \leq \overline\weight \|\Xi\|_1^{p/2} 
      \quad \text{for all $\Xi\in\smat$};
    \end{align}
 \item[(iii)] satisfies the following analogue of \eqref{eq.weirdass}:
    \begin{align}
      \label{eq.wweight}
      \E[\wweight(L^T\Xi L)+\wweight(R^T\Xi R)] \leq \kappa_p\wweight(\Xi) \quad \text{for all $\Xi\in\smat$}.
    \end{align}
  \end{itemize}
\end{lem}
\begin{proof}
  Relation \eqref{eq.wwhom} is an immediate consequence of the homogeneity of $\weight$
  and the fact that
  \begin{align*}
    \Xi = \sum_{j=1}^n\tau_j\xi_j\xi_j^T 
    \quad \Longleftrightarrow \quad 
    \lambda\Xi = \sum_{j=1}^n \tau_j\big(\sqrt\lambda\xi_j\big)\big(\sqrt\lambda\xi_j\big)^T.
  \end{align*}
  To prove the \emph{second} inequality in \eqref{eq.wwcompatible}, 
  let $\Xi\in\smat$ be given with spectral decomposition 
  \begin{align}
    \label{eq.specky}
    \Xi = \sum_{j=1}^d \lambda_j \ee_j\ee_j^T .
  \end{align}
  Choosing $\hat\xi_j=\sqrt{|\lambda_j|}\ee_j$ and $\hat\tau_j=\sgn\lambda_j$, 
  it follows by definition of $\wweight$ as infimum
  and in combination with \eqref{eq.boundass} that
  \begin{align*}
    \wweight(\Xi) \leq \sum_{j=1}^d \weight(\hat\xi_j) \leq \overline\weight\sum_{j=1}^d |\lambda_j|^{p/2} 
    \leq \overline\weight\bigg(\sum_{j=1}^d|\lambda_j| \bigg)^{p/2} 
    = \overline\weight \|\Xi\|_1^{p/2}.
  \end{align*}
  To obtain the first inequality in \eqref{eq.wwcompatible},
  choose $\Xi\in\smat$ and $\epsilon>0$,
  and let $\xi_1,\ldots,\xi_d\in\setR^d$ be such that
  \begin{align}
    \label{eq.epsdecompose}
    \Xi = \sum_{j=1}^d \tau_j\xi_j\xi_j^T \quad \text{and} \quad \sum_{j=1}^d\weight(\xi_j)<\wweight(\Xi)+\epsilon.
  \end{align}
  In combination with the spectral decomposition \eqref{eq.specky},
  one has
  \begin{align*}
    \|\Xi\|_1 = \sum_{k=1}^d |\lambda_k| = \sum_{k=1}^d |\ee_k^T\Xi\ee_k| 
    = \sum_{k=1}^d \bigg| \sum_{j=1}^d\tau_j(\ee_k^T\xi_j)^2 \bigg|
    \leq \sum_{j,k=1}^d (\ee_k^T\xi_j)^2 = \sum_{j=1}^d |\xi_j|^2.
  \end{align*}
  Then, using Jensen's inequality --- recall that $p/2>1$ --- and the boundedness \eqref{eq.boundass},
  \begin{align*}
    \|\Xi\|_1^{p/2} \leq d^{p/2-1} \sum_{j=1}^d |\xi_j|^p 
    \leq d^{p/2-1} \sum_{j=1}^d \weight(\xi_j) 
    = d^{p/2-1} \big(\wweight(\Xi)+\epsilon\big).
  \end{align*}
  Taking the limit $\epsilon\downarrow0$ shows \eqref{eq.wwcompatible}.
  Furthermore, again under hypothesis \eqref{eq.epsdecompose},
  one has for every matrix pair $(L,R)$ that
  \begin{align*}
    L^T\Xi L = \sum_{j=1}^d \tau_j(L^T\xi_j)(L^T\xi_j)^T, \quad 
    R^T\Xi R = \sum_{j=1}^d \tau_j(R^T\xi_j)(R^T\xi_j)^T.
  \end{align*}
  And consequently,
  \begin{align*}
    \E\big[\wweight(L^T\Xi L)+\wweight(R^T\Xi R)\big] 
    \leq \sum_{j=1}^d \E\big[\weight(L^T\xi_j)+\weight(R^T\xi_j)\big]
    \leq  \kappa_p\sum_{j=1}^d\weight(\xi_j) = \kappa_p\wweight(\Xi)+\kappa_p\epsilon,
  \end{align*}
  proving \eqref{eq.wweight} in the limit $\epsilon\downarrow0$.
\end{proof}

\section{The stationary problem}
\label{sct.stationary}
In this section, we prove --- under Assumptions 1 to 3 --- the existence 
of a non-degenerate distribution $\mu_\infty\in\cls_p(\setR^d)$ satisfying the stationary Boltzmann-equation $\mu_\infty=Q_+[\mu_\infty]$.
Also, we analyze the regularity of $\mu_\infty$ and investigate the finiteness of its moments above the second.
The existence theory is based on the ansatz \eqref{eq.scalemix} below;
we defer the proof of uniqueness for the stationary solution to Section \ref{sct.dynamic}.

\subsection{Existence of a stationary distribution}
Throughout this section, we work under the hypothesis that Assumption \ref{ass.3new} holds with some $p\in(2,3)$;
by Lemma \ref{lem.prime} above, there is no loss of generality in doing so.
Next, recall that the Fourier-Stieltjes transform $\hat\mu$ of a probability measure on $\setR^d$ 
is defined by
\[
\hat\mu (\xi) = \int_{\setR^d} e^{i\xi^T v} \mu(\dd v),
\]
for every $\xi\in\setR^d$.
Taking $\psi(v)=\exp(i\xi\cdot v)$ as a test function in \eqref{eq.gain},
we obtain the following Fourier representation $\widehat Q_+$ of the gain operator:
\begin{align}
  \label{eq.qfourier}
  \widehat Q_+[\hat\mu_1,\hat\mu_2] = \E\big[\hat\mu_1(L^T\xi)\hat\mu_2(R^T\xi)].
\end{align}
We prove the existence of a solution $\Psi=\hat\mu_\infty$ to the stationary equation
\begin{align}
  \label{eq.stationary}
  \Psi(\xi) = \E\big[\Psi(L^T\xi)\Psi(R^T\xi) \big]
\end{align}
which is a scale mixture of Gaussians,
\begin{align}
  \label{eq.scalemix}
  \Psi(\xi) = \E\big[\exp\big(-\half \xi^TS\xi\big) \big].
\end{align}
Here $S$ is a random matrix in $\smat_+$, whose distribution is to be determined.
\begin{rmk}
  The representation \eqref{eq.scalemix} is equivalent to \eqref{eq.fmix} in Theorem \ref{thm.main}.
  Assuming that $S$ is positive definite a.s., then \eqref{eq.fmix} takes the probably more familiar form
  \begin{align}
    \label{eq.smix}
    \frac{\dd\mu_\infty}{\dd v}  (v) = \E\left[\frac1{\operatorname{det}\sqrt{2\pi S}}\exp\Big(-\frac12 v^TS^{-1}v\Big)\right]
  \end{align}
  of a (non-isotropic) scale mixture of Gaussians.
\end{rmk}
\begin{lem}
  \label{lem.aux}
  Suppose that $\nu_\infty$ is a probability measure on $\smat_+$ with the following property:
  if $S$, $S'$ and $S''$ are random matrices with identical distribution $\nu_\infty$,
  which are mutually independent and also independent of $(L,R)$,
  then
  \begin{align}
    \label{eq.aux}
    S \distreq LS'L^T + RS''R^T.
  \end{align}
  Then \eqref{eq.scalemix} defines a solution to \eqref{eq.stationary}.
\end{lem}
\begin{proof}
  Let $\sigalg$ be the $\sigma$-algebra generated by $(L,R)$.
  Define $\Psi$ as in \eqref{eq.scalemix}.
  Then, using \eqref{eq.aux},
  \begin{align*}
    \E\big[\Psi(L^T\xi)\Psi(R^T\xi) \big]
    &= \E\Big[ \E\big[\exp\big(-\half (L^T\xi)^TS'(L^T\xi)\big)\big|\sigalg\big] 
    \E\big[\exp\big(-\half (R^T\xi)^TS''(R^T\xi)\big)\big|\sigalg\big] \Big] \\
    &= \E\big[ \exp\big(-\half \xi^T(LS'L^T+RS''R^T)\xi\big) \big] \\
    &= \E\big[ \exp\big(-\half \xi^TS\xi\big)\big] = \Psi(\xi)
  \end{align*}
  holds for every $\xi\in\setR^d$.
\end{proof}
The rest of this section is devoted to proving
\begin{thm}
  \label{thm.exuniq}
  Under Assumptions 1--3, 
  there exists precisely one measure $\nu_\infty\in\cls_p(\smat)$ which satisfies the hypotheses of Lemma \ref{lem.aux}.
  Consequently, the probability measure $\mu_\infty$ on $\setR^d$ associated to $\nu_\infty$ via \eqref{eq.scalemix}
  is a solution to the stationary equation $Q_+[\mu]=\mu$ in the class $\cls_p(\setR^d)$.
  In particular, $\mu_\infty$ has unit temperature and finite $p$th moment.
\end{thm}
We shall obtain $\nu_\infty$ as the fixed point of a map $T$ defined as follows:
Given two probability measures $\nu',\nu''$ on $\smat_+$, 
let $S'$ and $S''$ be random matrices in $\smat_+$ with distribution $\nu'$ and $\nu''$, respectively,
that are independent of each other, and independent of $(L,R)$.
Define
\begin{align*}
  T[\nu',\nu''] = \law[LS'L^T+RS''R^T].
\end{align*}
Clearly, any solution $\nu_\infty$ of $\nu=T[\nu,\nu]$ verifies the hypotheses of Lemma \ref{lem.aux}.
\begin{lem}
  \label{lem.trace2} 
  Let Assumptions 1 and 2 hold.
  If $\nu'$ and $\nu''$ belong to $\cls_p(\smat_+)$, then $T[\nu',\nu'']\in\cls_p(\smat_+)$.
\end{lem}
\begin{proof}
  By linearity and cyclicity of the trace, by Assumption 1, and since $(L,R)\distreq(R,L)$,
  \[
  \begin{split}
    \tr\E[LS'L^T+RS''R^T]
    &=\E[\tr(LS'L^T+RS''R^T)] =\tr(\E[L^TL]\E[S']+\E[R^TR]\E[S''])\\
    &=\frac{1}{2}\tr(\E[S'+S''])=d.
  \end{split}
  \]
  Non-negativity is preserved, 
  since $v^TS'v\geq0$ for all $v\in\setR^d$ clearly implies that $w^T(LS'L^T)w=(L^Tw)^TS'(L^Tw)\geq0$ for every $w\in\setR^d$,
  and similarly for $S''$.
  Finally, note that $\E[\tr(S')^{p/2}]<+\infty$ by hypothesis,
  that $\E[\|L^T\|_2^p]<+\infty$ follows from Assumption \ref{ass.3new},
  and that hence
  \begin{align*}
    \E\big[\big(\tr(LS'L^T)\big)^{p/2}\big]\leq \E\big[\big(\tr S'\big)^{p/2}\big] \E\big[\|L^T\|_2^p\big] < +\infty 
  \end{align*}
  by inequality \eqref{eq.submult}, and by independence of $S'$ and $L$.
  The respective estimate on $\tr(RS''R^T)$ follows in the same way,
  thus concluding the proof.
\end{proof}
In order to prove existence and uniqueness of a solution to $\nu=T[\nu,\nu]$,
we are going to show that $T$ is contractive with respect to the following distance
between elements $\nu_1,\nu_2\in\cls_p(\smat_+)$:
\begin{align}
  \label{eq.ourmetric}
  \fnorm_\wweight(\nu_1,\nu_2) := \delta_\wweight(\nu_1,\nu_2) + \alpha \big\| \E[S_1-S_2] \big\|_\infty.
\end{align}
Here $\alpha>0$ is a parameter (to be specified later),
$S_1,S_2$ are independent random matrices with respective distributions $\nu_1,\nu_2$,
and
\begin{align*}
  \delta_\wweight(\nu_1,\nu_2):= \sup_{0\neq\Xi\in\smat}  \wweight(\Xi)^{-1} 
  \big|\hat\nu_1(\Xi)-\hat\nu_2(\Xi) - i\tr\big(\Xi^T\E[S_1-S_2]\big) \big| .
\end{align*}
$\hat\nu_i(\Xi)$ denotes the characteristic functions of $\nu_i$ at $\Xi\in\setR^{d\times d}$,
which can conveniently be written as follows:
\begin{align}
   \label{eq.matrixfourier}
  \hat\nu_i(\Xi) = \E\big[\exp(-i\tr(\Xi^TS_i)\big],
\end{align}
with $S_i$ being a random matrix of distribution $\nu_i$.
In fact, it suffices to evaluate $\hat\nu_i(\Xi)$ on matrices $\Xi\in\smat$,
since if $\Xi^T=-\Xi$, it follows by the cyclicity of the trace and from $S_i^T=S_i$ a.s. 
that 
\begin{align*}
  -\tr(\Xi S_i) = \tr(\Xi^TS_i) = \tr((S_i^T\Xi)^T)= \tr(S_i\Xi)= \tr(\Xi S_i),
\end{align*}
and hence equals zero.
\begin{lem}
  \label{lem.fnorm}
  The distance $\fnorm_\wweight$ defines a metric on $\cls_p(\smat_+)$
  with the following convexity property:
  If $\nu_1=s\nu_1'+(1-s)\nu_1''$ and $\nu_2=s\nu_2'+(1-s)\nu_2''$ are two convex combinations 
  of measures in $\cls_p(\setR^d)$, with $s\in[0,1]$, 
  then
  \begin{align}
    \label{eq.fnormconvex}
    \fnorm_\wweight(\nu_1,\nu_2) \leq s \fnorm_\wweight(\nu_1',\nu_2') + (1-s) \fnorm_\wweight(\nu_1'',\nu_2'').
  \end{align}
  Moreover, $\fnorm_\wweight$ provides a pointwise control on the Fourier transforms,
  \begin{align}
    \label{eq.fnormcontrol}
    \big| \hat\nu_1(\Xi)-\hat\nu_2(\Xi) \big| \le K\max\big( \|\Xi\|_1,\|\Xi\|_1^{p/2} \big) \,\fnorm_\wweight(\nu_1,\nu_2),
  \end{align}
  with a constant $K$ that only depends on $\alpha$ and $\wweight$.
  Finally, convergence $\fnorm_\wweight(\nu_k,\nu_*)\to0$ for a sequence $\nu_k$ in $\cls_p(\smat_+)$ 
  implies weak convergence of $\nu_k$ to $\nu_*$ and convergence of the expectation values.
\end{lem}
\begin{proof}
  The only difficulty here is to prove well-definiteness of $\fnorm_\wweight$:
  given $\nu_1,\nu_2\in\cls_p(\smat_+)$, we need to show that $\delta_\wweight(\nu_1,\nu_2)<\infty$.
  Using estimate \eqref{eq.taylor1} from the appendix in the definition \eqref{eq.matrixfourier} of $\hat\nu_1$, 
  we obtain (recall that we assume $p\in(2,3)$ throughout this section)
  \begin{align*}
    \big|\hat\nu_1(\Xi) - \big(1-i\Xi^T\E[S_1]\big) \big|
    & \leq \E\big[ \big| \exp(i\Xi^TS_1) - (1+i\Xi^TS_1) \big| \big] \\
    & \leq C_p \E\big[|\Xi^TS_1|^{p/2} \big] \leq C_p\|\Xi\|_1^{p/2}\E\big[\|S_1\|_\infty^{p/2}\big].
  \end{align*}
  And thus, by the triangle inequality and \eqref{eq.wwcompatible},
  \begin{align*}
    \delta_\wweight(\nu_1,\nu_2) 
    \leq C_p \E\big[\|S_1\|_\infty^{p/2}+\|S_2\|_\infty^{p/2}\big] \sup_{\Xi\neq0}\big(\wweight(\Xi)^{-1}\|\Xi\|_1^{p/2}\big)
    \leq C_p d^{p/2-1}\E\big[(\tr S_1)^{p/2}+(\tr S_2)^{p/2}\big].
  \end{align*}
  The last expression is finite by the definition of $\cls_p(\smat_+)$.

  The rest of the proof is straight-forward:
  Symmetry, non-negativity and the triangle inequality 
  follow directly from the definition of $\fnorm_\wweight$.
  The convexity property \eqref{eq.fnormconvex} is a consequence 
  of the respective individual convexities of $\delta_\wweight$ and $\|\cdot\|_\infty$.
  The estimate \eqref{eq.fnormcontrol} is obtained by observing that
  \begin{align*}
    \big|\hat\nu_1(\Xi)-\hat\nu_2(\Xi)\big|
    \le \wweight(\Xi) \fnorm_\wweight(\nu_1,\nu_2) + | \tr ( \Xi^T\E[S_1-S_2]) |
  \end{align*}
  follows directly from the definition of $\fnorm_\wweight$,
  that
  \begin{align*}
    | \tr ( \Xi^T\E[S_1-S_2]) | \le \|\Xi\|_1 \|\E[S_1-S-2]\|_\infty 
    \le \frac1\alpha \|\Xi\|_1 \fnorm_\wweight(\nu_1,\nu_2)
  \end{align*}
  holds by relation \eqref{eq.matrixholder} for matrix norms,
  and by recalling relation \eqref{eq.wwcompatible} between $\|\Xi\|_1$ and $\wweight(\Xi)$.
  Finally, if $\fnorm_\wweight(\nu_k,\nu_*)\to0$,
  then in particular $\E[S_k]\to\E[S_*]$,
  and the pointwise convergence $\hat\nu_k(\Xi)\to\hat\nu_\infty(\Xi)$ for every $\Xi\in\smat$ due to \eqref{eq.fnormcontrol}
  implies weak convergence of the measures $\nu_k$ to $\nu_*$.
\end{proof}
%
\begin{lem}
  \label{Lemma4} 
  Let Assumptions 1--3 hold.
  Let $(\nu_1',\nu_1'')$ and $(\nu_2',\nu_2'')$ be two pairs of measures, 
  each of which belongs to $\cls_p(\smat_+)$,
  and define $\nu_1=T[\nu_1',\nu_1'']$, $\nu_2=T[\nu_2',\nu_2'']$.
  Then
  \begin{equation}
    \label{eq.fnormcontract}
    \delta_\wweight(\nu_1,\nu_2) \leq \kappa_p \left\{
      \frac12\big(\delta_\wweight (\nu_1',\nu_2')+\delta_\wweight (\nu_1'',\nu_2'')\big)
      + d^{p/2} \big( \| \E[S_1'-S_2'] \|_\infty + \| \E[S_1''-S_2'']\|_\infty \big) \right\}.
  \end{equation}
  {where
  $S_i'$, $S_i''$ respectively, are independent random matrices with law $\nu_i'$, $\nu_i''$ respectively,
  for $i=1,2$. }
\end{lem}
\begin{proof}
  Fix some $\Xi\in\smat$.
  For $\nu_i=T[\nu_i',\nu_i'']$, we obtain the representations
  \begin{align*}
    \hat\nu_i(\Xi) &= \E\big[\exp\big(-i\tr\Xi^T(LS_i'L^T+RS_i''R^T)\big)\big] \\
    &= \E\Big[\E\big[\exp\big(-i\tr(L^T\Xi L)^TS_i'\big)\big|\sigalg\big]
    \E\big[\exp\big(-i\tr(R^T\Xi R)^TS_i''\big)\big|\sigalg\big]\Big] \\
    &= \E\big[\hat\nu_i'(L^T\Xi L)\hat\nu_i''(R^T\Xi R)\big].
  \end{align*}
  Now define
  \begin{align*}
    \Delta':=\E[S_1'-S_2'], \quad
    \Delta'':=\E[S_1''-S_2''],
  \end{align*}
  and observe that
  \begin{align*}
    \Delta:=\E[S_1-S_2] = \E[L\Delta'L^T+R\Delta''R^T]
  \end{align*}
  {where $S_i$ has law $\nu_i$.}
  By the triangle inequality, 
  using that characteristic functions are bounded in modulus by one,
  \begin{align*}
    \delta &:= \big| \hat\nu_1(\Xi)-\hat\nu_2(\Xi) - i\tr(\Xi^T\Delta) \big| \\
    & = \big|\E[\hat\nu_1'(L^T\Xi L)\hat\nu_1''(R^T\Xi R)-\hat\nu_2'(L^T\Xi L)\hat\nu_2''(R^T\Xi R) -i\tr\Xi^T(L\Delta'L^T+R\Delta''R^T)]\big| \\
    & \leq \E\big[|\hat\nu_1'(L^T\Xi L)-\hat\nu_2'(L^T\Xi L)-i\tr(L^T\Xi L)^T\Delta'|\big] + \E[ |1-\hat \nu_1''(R^T\Xi R)| |\tr(L^T\Xi L)^T\Delta'| ] \\
    & \quad + \E\big[|\hat\nu_1''(R^T\Xi R)-\hat\nu_2''(R^T\Xi R)-i\tr(R^T\Xi R)^T\Delta''|\big] + \E[ |1-\hat \nu_2'(L^T\Xi L)| |\tr(R^T\Xi R)^T\Delta''| ] \\
    & \leq \E[\wweight(L^T\Xi L)] \sup_{\Xi'\neq0}\left(\frac{|\hat\nu_1'(\Xi')-\hat\nu_2'(\Xi')-i\tr(\Xi')^T\Delta'|}{\wweight(\Xi')}\right) \\
    & \qquad + \E[\wweight(R^T\Xi R)] \sup_{\Xi''\neq0} \left(\frac{|\hat\nu_1''(\Xi'')-\hat\nu_2''(\Xi'')-i\tr(\Xi'')^T\Delta''|}{\wweight(\Xi'')}\right) \\
    & \quad + \E\big[\|R^T\Xi R\|_1^{p/2-1} |\tr(L^T\Xi L)^T\Delta'| \big] \sup_{\Xi''\neq0}\left(\frac{|1-\hat \nu_1''(\Xi'')|}{\|\Xi''\|_1^{p/2-1}}\right) \\
    & \qquad + \E\big[\|L^T\Xi L\|_1^{p/2-1} |\tr(R^T\Xi R)^T\Delta''| \big] \sup_{\Xi'\neq0}\left(\frac{|1-\hat \nu_2'(\Xi')|}{\|\Xi'\|_1^{p/2-1}}\right).
  \end{align*}
  The first two terms are estimated further using $\E[\wweight(L^T\Xi L)]=\E[\wweight(R^T\Xi R)]\leq(\kappa_p/2)\wweight(\Xi)$,
  which follows from \eqref{eq.wweight} and the symmetry $(L,R)\distreq(R,L)$.
  For estimation of the last two terms,
  we apply the rule \eqref{eq.matrixholder}
  to the expressions containing the trace,
  and then Young's inequality, see \eqref{eq.young}, to the product inside the expectation.
  Altogether, this yields
  \begin{align*}
    & \delta \leq \frac{\kappa_p\wweight(\Xi)}2
    \left[\sup_{\Xi'\neq0}\left(\frac{|\hat\nu_1'(\Xi')-\hat\nu_2'(\Xi')+i\tr(\Xi')^T\Delta'|}{\wweight(\Xi')}\right)
      +\sup_{\Xi''\neq0} \left(\frac{|\hat\nu_1''(\Xi'')-\hat\nu_2''(\Xi'')+i\tr(\Xi'')^T\Delta''|}{\wweight(\Xi'')}\right)\right] \\
    & + \frac12 \E\big[ \|L^T\Xi L\|_1^{p/2} + \|R^T\Xi R\|_1^{p/2} \big]
    \left[ \|\Delta'\|_\infty \sup_{\Xi''\neq0}\left(\frac{|1-\hat \nu_1''(\Xi'')|}{\|\Xi''\|_1^{p/2-1}}\right)
      + \|\Delta''\|_\infty  \sup_{\Xi'\neq0}\left(\frac{|1-\hat \nu_2'(\Xi')|}{\|\Xi'\|_1^{p/2-1}}\right) \right].
  \end{align*}
  The second terms needs more estimates.
  Observe that by \eqref{eq.wwcompatible} and \eqref{eq.wweight},
  one has
  \begin{align*}
    \E\big[ \|L^T\Xi L\|_1^{p/2} + \|R^T\Xi R\|_1^{p/2} \big] 
    \leq d^{p/2-1} \E\big[ \wweight(L^T\Xi L)+\wweight(R^T\Xi R) \big]
    \leq d^{p/2-1}\kappa_p \wweight(\Xi).
  \end{align*}
  It remains to bound the supremum involving $|1-\nu_2'(\Xi')|$.
  Since $\nu_2' \in \cls_p(\smat_+)$, one has $\tr\E[S_2']=d$.
  Invoking \eqref{eq.matrixholder},
  and using that $\|\Sigma\|_\infty\leq \tr(\Sigma)$ for every $\Sigma\in\smat_+$,
  one finds that
  \begin{align*}
    |1-\hat\nu_2'(\Xi)| & \leq \int_0^1 \bigg|\frac{d}{ds}\hat\nu_2'(s\Xi)\bigg|\dd s 
    \leq \int_0^1 \left|\E\big[-i\tr(\Xi^TS_2')\exp\big(-is\tr(\Xi^TS_2')\big)\big] \right|\dd s \\
    & \leq \|\Xi\|_1\E\big[\|S_2'\|_\infty\big] \leq \|\Xi\|_1\E\big[\tr S_2'\big] = d \|\Xi\|_1 . 
  \end{align*}
  In combination with $|1-\hat\nu_2'(\Xi')|\leq 2$ for all $\Xi'\in\smat$,
  one derives
  \[
  \sup_{\Xi'\neq0} \frac{|1-\hat\nu_2'(\Xi')|}{\|\Xi'\|_1^{p/2-1}} 
  \leq \sup_{\Xi'\neq0} \min\{ 2\|\Xi'\|_1^{1-p/2}  ,d\|\Xi'\|_1^{2-p/2}   \} = 2^{2-p/2}d^{p/2-1} \leq 2d,
  \]
  and the analogous estimate for the term involving $\hat\nu_1''(\Xi'')$.
  Putting everything together, \eqref{eq.fnormcontract} follows.
  Note that we have implicitly used our additional assumption $2<p<3$.
\end{proof}
\begin{lem}
  \label{lem.fwcontract}
   Under the same assumptions and with the same notations as in Lemma \ref{Lemma4},
   it follows that
   \begin{align}
     \label{eq.fwcontract}
     \fnorm_\wweight(\nu_1,\nu_2) \leq \frac\lambda2 \big( \fnorm_\wweight(\nu_1',\nu_2') + \fnorm_\wweight(\nu_1'',\nu_2'') \big)
   \end{align}
   with the constant
   \begin{align}
     \label{eq.deflambda}
     \lambda = \max\big( \kappa_p,\kappa+2\alpha^{-1}d^{p/2}\kappa_p \big),
   \end{align}
   which is smaller than one if $\alpha$ is large enough.
\end{lem}
\begin{proof}
  It suffices to combine equation \eqref{eq.fnormcontract} above with the observation that
  \begin{align*}
    \big\|\E[S_1-S_2] \big\|_\infty 
    & = \big\| \E[L(S_1'-S_2')L^T+R(S_1''-S_2'')R^T] \big\|_\infty \\
    & \leq \frac12 \big( \big\| \tq(\E[S_1'-S_2']) \big\|_\infty + \big\| \tq(\E[S_1''-S_2'']) \big\|_\infty \big) \\
    & \leq \frac\kappa2 \big( \|\E[S_1'-S_2']\|_\infty +  \|\E[S_1''-S_2'']\|_\infty \big).
  \end{align*}
  Here we have used the triangle inequality for $\|\cdot\|_\infty$,
  the symmetry $(L,R)\distreq(R,L)$,
  and Assumption \ref{ass.2} in combination with the fact that $\tr\E[S_1']=\tr\E[S_2']=d$.
\end{proof}
\begin{lem}
  \label{lem.tight}
  Let Assumptions 1--3 hold.
  There exists a constant $C$ such that the following is true:
  For arbitrary $\nu',\nu''\in\cls_p(\smat_+)$,
  and if $S$, $S'$ and $S''$ independent random matrices with law, $\nu:=T[\nu',\nu'']$, $\nu'$ and $\nu''$, 
  respectively,
  then
 \begin{equation}
    \label{eq.tight}
    \sup_{\xi\neq0} \frac{ \E[(\xi^T S \xi)^{p/2}]}{\weight(\xi)} 
    \leq \frac{\kappa_p}{2} \left( \sup_{\xi'\neq0} \frac{\E[({\xi'}^TS'\xi')^{p/2}]}{\weight(\xi')}
      +\sup_{\xi''\neq0} \frac{ \E[({\xi''}^T S'' \xi'')^{p/2}]}{\weight(\xi'')}\right) + C. 
  \end{equation}
\end{lem}
\begin{proof}
  The starting point is an application of inequality \eqref{eq.myineq}:
  \begin{align*}
    \sup_{\xi\neq0}& \frac{\E\big[(\xi^TS\xi)^{p/2}\big]}{\weight(\xi)}
    =\sup_{\xi\neq0} \frac{\E\Big[\E\big[\big(\xi^T(LS'L^T+RS''R^T)\xi\big)^{p/2}\big|\sigalg\big]\Big]}{\weight(\xi)} \\
    & \leq \sup_{\xi \neq 0} \frac1{\weight(\xi)} 
    \Big \{ \E \big[ \E[(\xi^T LS'L^T\xi)^{p/2}|\sigalg] +  \E[(\xi^TRS''R^T\xi)^{p/2}|\sigalg] \big ] \\ 
    &\quad + C_p  \E\big[\E\big[(\xi^TLS'L^T\xi)(\xi^TRS''R^T\xi)^{p/2-1}
    + (\xi^T LS'L^T\xi)^{p/2-1}(\xi^T RS''R^T\xi)\big|\sigalg\big] \big]\Big \}.
  \end{align*}
 The first two terms inside the expectation are estimated directly using Assumption \ref{ass.3new}:
  \begin{align*}
    & \sup_{\xi \neq0} \frac{1}{\weight(\xi)} \Big \{ 
    \E \big[ \E[(\xi^T LS'L^T \xi )^{p/2}|\sigalg] + \E[(\xi^T RS''R^T\xi)^{p/2}|\sigalg] \big]\Big \} \\
    & \leq \sup_{\xi \neq 0}  \E \bigg[
    \frac{\weight(L^T\xi)}{\weight(\xi)} \sup_{\xi' \neq 0} \frac{\E[|{\xi'}^TS'\xi'|^{p/2}]}{\weight(\xi')} + 
    \frac{\weight(R^T\xi)}{\weight(\xi)} \sup_{\xi'' \neq 0} \frac{\E[|{\xi''}^TS''\xi''|^{p/2}]}{\weight(\xi'')} \bigg] \\
    & \leq \frac{\kappa_p}{2} \Big \{  \sup_{\xi' \neq 0} \frac{\E[({\xi'}^T S'\xi')^{p/2}]}{\weight(\xi')} 
    + \sup_{\xi'' \neq 0} \frac{ \E[({\xi''}^TS''\xi'')^{p/2}]}{\weight(\xi'')} \Big \} .
  \end{align*}
  To estimate the remainder terms, 
  we apply Young's inequality \eqref{eq.young} and inequality \eqref{eq.submult2},
  \begin{align}
    \nonumber
    & \E\Big[\E\big[(\xi^TLS'L^T\xi)(\xi^TRS''R^T\xi)^{p/2-1}\big|\sigalg\big]\Big]
    \leq \E\big[ \|S'\|_\infty \big]\E\big[ \|S''\|_\infty^{p/2-1} \big] \E\big[|L^T\xi|^2|R^T\xi|^{p-2} \big] \\
    \label{eq.thisremainder}
    & \qquad \leq \E[\tr S']\E[(\tr S'')^{p/2-1}]\E\left[ \frac{p-2}{p}|L^T\xi|^p + \frac2p |R^T\xi|^p \right].
  \end{align}
  Finally, Jensen's inequality gives $\E[(\tr S'')^{p/2-1}]\leq\E[\tr S'']^{p/2-1}=d^{p/2-1}$;
  recall our additional assumption that $2<p<3$.
  A similar estimate is derived for the other remainder term;
  in summary, we obtain
  \begin{align*}
    &\E\Big[\E\big[(\xi^TLS'L^T\xi)(\xi^TRS''R^T\xi)^{p/2-1}
    + (\xi^T LS'L^T\xi)^{p/2-1}(\xi^T RS''R^T\xi)\big|\sigalg\big] \Big] \\
    & \qquad \leq d^{p/2}\kappa_p\E[\weight(L^T\xi)+\weight(R^T\xi)] \leq d^{p/2}\kappa_p\weight(\xi).
  \end{align*}
  This means that \eqref{eq.tight} holds with $C=d^{p/2}\kappa_pC_p$.
\end{proof}
\begin{proof}[Proof of Theorem \ref{thm.exuniq}]
  Choose some $\nu_0\in\cls_p(\smat_+)$, 
  and define inductively the iterates $\nu_m=T[\nu_{m-1},\nu_{m-1}]$ for all $m\geq1$.
  Denote by $S_0,\,S_1,\,S_2,\ldots$ mutually independent random matrices with
  respective distributions $\nu_0,\,\nu_1,\,\nu_2,\ldots$
  Iteration of estimate \eqref{eq.tight} gives
  \begin{align*}
    \sup_{\xi\neq0}\frac{\E\big[(\xi^TS_m\xi)^{p/2}\big]}{\omega(\xi)} 
    \leq \kappa_p^m \sup_{\xi\neq0}\frac{\E\big[(\xi^TS_0\xi)^{p/2}\big]}{\omega(\xi)} + \frac{1-\kappa_p^m}{1-\kappa_p}C
  \end{align*}
  for every $m\geq1$.
  This implies the $m$-uniform bound
  \begin{align}
    \label{eq.verytight}
    \E\big[\|S_m\|_\infty^{p/2} \big] \leq \overline\weight \E\big[\|S_0\|_\infty^{p/2} \big] + \frac{C\overline\weight}{1-\kappa_p},
  \end{align}
  which means that the $\nu_m$ form a tight family of probability measures on $\smat_+$.
  Consequently, there exists a subsequence $\nu_{m'}$ that converges weakly to a limit $\nu_\infty\in\cls_p(\smat_+)$.
  
  The next step is to identify the limit $\nu_\infty$ as a solution to the problem \eqref{eq.aux},
  i.e., we need to prove that $\nu_\infty=T[\nu_\infty,\nu_\infty]$.
  This is easy on the level of the Fourier transforms $\hat\nu_m$:
  we know that
  \begin{align}
    \label{eq.nu1}
    \hat\nu_{m'+1}(\Xi) = \E[\hat\nu_{m'}(L^T\Xi L)\hat\nu_{m'}(R^T\Xi R)]\qquad (\Xi\in\smat)
  \end{align}
  for all $m\geq1$, and we need to conclude that
  \begin{align}
    \label{eq.nu2}
    \hat\nu_\infty(\Xi) = \E[\hat\nu_\infty(L^T\Xi L)\hat\nu_\infty(R^T\Xi R)] \qquad (\Xi\in\smat).
  \end{align}  
  Weak converges implies in particular that
  the Fourier transforms $\hat\nu_{m'}$ converge pointwise to the Fourier transform $\hat\nu_\infty$ of the limit.
  Passing to the limit $m'\to\infty$ on the right-hand side of \eqref{eq.nu1},
  we obtain --- using dominated convergence --- the right-hand side of \eqref{eq.nu2},
  for every $\Xi\in\smat$.

  On the other hand, iteration of estimate \eqref{eq.fwcontract} provides the bound
  \begin{align*}
    \fnorm_\wweight(\nu_{m'+1},\nu_{m'}) \leq \lambda^{m'}\fnorm_\wweight(\nu_1,\nu_0).
  \end{align*}
  We assume that $\alpha>0$ has been chosen large enough so that $\lambda<1$ in \eqref{eq.deflambda}.
  By estimate \eqref{eq.fnormcontrol},
  it follows that 
  \begin{align*}
    |\hat\nu_{m'+1}(\Xi)-\hat\nu_{m'}(\Xi)| \leq K \big( \|\Xi\|_1,\|\Xi\|_1^{p/2}\big) \fnorm_\wweight(\nu_{m'+1},\nu_{m'}) \to 0
  \end{align*}
  at every $\Xi\in\smat$.
  So $\hat\nu_{m'+1}$ converges pointwise to the same limit as $\hat\nu_{m'}$, i.e., to $\hat\nu_\infty$.
  This allows to pass to the limit also on the left-hand side in \eqref{eq.nu1}.

  Finally, we show that $\nu_\infty$ is the only solution to \eqref{eq.aux}.
  If $\nu_\infty'$ is another solution, i.e., $\nu_\infty'=T[\nu_\infty',\nu_\infty']$,
  then
  \begin{align*}
    \fnorm_\wweight(\nu_\infty,\nu_\infty') \leq \lambda \fnorm_\wweight(\nu_\infty,\nu_\infty').
  \end{align*}
  Since $\lambda<1$, this means that $\fnorm_\wweight(\nu_\infty,\nu_\infty')=0$, and so $\nu_\infty=\nu_\infty'$.
\end{proof}

\subsection{Symmetries and moments of the stationary state}
Having proven the existence of stationary state $\mu_\infty$,
we now turn to study some of its properties.
Our first result is that spatial symmetries in the distribution of the coefficient matrices $L$ and $R$
induce symmetries of $\mu_\infty$.
\begin{prp}
  \label{cor.symmetry} 
  Let $\Gamma\subset\so(\setR^d)$ be a matrix subgroup that leaves $(L,R)$ invariant under conjugation,
  i.e., for any $\Theta\in\Gamma$,
  \begin{align}
    \label{eq.lrsymmetry}
    (\Theta L\Theta^{-1},\Theta R\Theta^{-1}) \distreq (L,R) .
  \end{align}
  Then the fixed point $\mu_\infty$ is $\Gamma$-invariant, i.e., $\mu_\infty\circ\Theta^{-1}=\mu_\infty$ for every $\Theta\in\Gamma$.
  In particular, $\mu_\infty$ is always point symmetric, 
  i.e., $\mu_\infty(V)=\mu_\infty(-V)$ for all measurable sets $V\subset\setR^d$.
\end{prp}
\begin{proof}
  Let $X$, $X'$ and $X''$ be independent random vectors with distribution $\mu_\infty$,
  and let $\Theta\in\Gamma$.
  Then, by definition of $\mu_\infty$ as fixed point of $Q_+$, 
  by linearity of $\Theta$,
  and by the symmetry relation \eqref{eq.lrsymmetry},
  \begin{align*}
    \Theta X \distreq \Theta(LX'+RX'') = \Theta L\Theta^{-1}\Theta X'+\Theta R\Theta^{-1}\Theta X'' \distreq L\Theta X'+R\Theta X''.
  \end{align*}
  This means that also $\law(\Theta X)=\mu_\infty\circ\Theta^{-1}$ defines a fixed point of $Q_+$.
  Clearly, $\mu_\infty\circ\Theta^{-1}\in\cls_q(\setR^d)$ since $\mu_\infty\in\cls_q(\setR^d)$, 
  and $|\Theta v|=|v|$ because $\Theta\in\so(\setR^d)$.
  But since $Q_+$'s fixed point is unique in $\cls_q$, the first claim follows.

  The second claim is a trivial consequence of the fact 
  that one can always choose $\Gamma=\{\eins,-\eins\}$.
  Indeed, one even has $(\Theta L\Theta^{-1},\Theta R\Theta^{-1})=(L,R)$ for $\Theta=\pm\eins$.
\end{proof}
The following two propositions are concerned with the finiteness of moments of $\mu_\infty$.
We emphasize that we do no longer assume that $p<3$, 
but we allow for general $p\ge2$ in Assumption \ref{ass.3new}.
\begin{prp}
  The $p$th absolute moment of $\mu_\infty$ is finite.
\end{prp}
\begin{proof}
  This requires just a minor modification of the proof of Lemma \ref{lem.tight}.
  Namely, we estimate the remainder term in \eqref{eq.thisremainder} in a slightly different way,
  using
  \begin{align*}
    \E\big[(\tr S'')^{p/2-1}\big] \le \E\big[(\tr S'')^{p/2}\big]^{1-2/p} 
    \le d^{p/2-1}\overline\weight\sup_{\xi''\neq0}\left(\frac{\E\big[\big((\xi'')^TS''\xi''\big)^{p/2}\big]}{\weight(\xi'')}\right)^{1-2/p}.
  \end{align*}
  This means that \eqref{eq.tight} is replaced by a different estimate,
  \begin{align*}
    \sup_{\xi\neq0}\frac{\E\big[\big(\xi^TS\xi\big)^{p/2}\big]}{\weight(\xi)}
    \le F\left( \sup_{\xi'\neq0}\frac{\E\big[\big((\xi')^TS'\xi'\big)^{p/2}\big]}{\weight(\xi')},
        \sup_{\xi''\neq0}\frac{\E\big[\big((\xi'')^TS''\xi''\big)^{p/2}\big]}{\weight(\xi'')} \right),
  \end{align*}
  with the nonlinear function
  \begin{align*}
    F(z_1,z_2) = \frac{\kappa_p}2 (z_1+z_2) + d^{p/2}\kappa_p\overline\weight \big( z_1^{1-2/p}+z_2^{1-2/p}\big).
  \end{align*}
  The sequence of iterates $\nu_m$ in the proof of Theorem \ref{thm.exuniq} thus satisfies
  the following variant of \eqref{eq.verytight},
  \begin{align*}
    \E\big[ \|S_m\|_\infty^{p/2} \big] \le \overline\weight 
    \max\Big( \|S_0\|_\infty^{p/2}, \left(\frac{d^{p/2}\kappa_p\overline\weight}{1-\kappa_p}\right)^{2/p} \Big)
  \end{align*}
  which shows that the $p/2$-th moment of $\nu_m$ is $m$-uniformly bounded.
  Thus, the $p/2$-th moment of $S_\infty$ is finite.
  Now, by definition of $\mu_\infty$ in \eqref{eq.scalemix} as a scale mixture,
  it follows that
  \begin{align*}
    \int_{\setR^d} |v|^p\mu_\infty(\dd v) 
    = \int_{\setR^d} \E\big[|w^TSw|^{p/2}\big] \frac{e^{-|w|^2/2}}{(2\pi)^{d/2}}\dd w
    \le \E\big[\|S\|_\infty^{p/2}\big] \int_{\setR^d} |w|^p\frac{e^{-|w|^2/2}}{(2\pi)^{d/2}}\dd w,
  \end{align*}
  which is finite.
\end{proof}
\begin{prp}
  \label{prp.moments}
  Suppose that, for some $s>p>2$,
  \begin{align}
    \kappa^*_s := \inf_{|\ee|=1}\E\big[|L\ee|^s \big] + \inf_{|\ee|=1}\E\big[|R\ee|^s \big] > 1.
  \end{align}
  Then the $s$th absolute moment of $\mu_\infty$ is infinite.

  Moreover, if $(L,R)$ is invariant under conjugation of the full $\so(\setR^d)$,
  then
  \begin{align}
    \label{eq.bbstar}
    \kappa^*_s  = \E[ |L\ee_0|^s ] + \E[ |R\ee_0|^s ],
  \end{align}
  where $\ee_0\in\sphere^{d-1}$ is an arbitrarily chosen unit vector.
\end{prp}
\begin{proof}
  Let $X$, $X'$ and $X''$ be a independent random vectors with distribution $\mu_\infty$.
  By definition of $\mu_\infty$ as fixed point of $Q_+$,
  and by the point symmetry stated in Corollary \ref{cor.symmetry} above,
  \begin{align*}
    X' \distreq LX'+RX'' \distreq LX'-RX'' .
  \end{align*}
  Now, to prove divergence of the $s$th moment, 
  assume that, on the contrary, $\E[|X|^s]<\infty$.
  Then the following calculation would be legitimate:
  \begin{align*}
    & \E\big[ |X|^s \big]
    = \frac12 \Big\{ \E\big[ |LX'+RX''|^s \big]
    + \E\big[ |LX'-RX''|^s \big]\Big\} \\
    & = \E\Big[ \frac12\Big\{ \big(|LX'|^2+|RX''|^2+2(X')^TL^TRX'' \big)^{s/2} 
    + \big(|LX'|^2+|RX''|^2-2(X')^TL^TRX'' \big)^{s/2} \Big\} \Big] \\
    & \geq \E\big[\big(|LX'|^2+|RX''|^2\big)^{s/2}\big].
  \end{align*}
  The last inequality follows directly from the convexity of $z\mapsto z^{s/2}$.
  Observe further that since $s>2$, 
  one has $(a^2+b^2)^{s/2}\geq a^s+b^s$ for arbitrary non-negative real numbers $a$ and $b$.
  Consequently, with random unit vectors $\ee'=X'/|X'|$ and $\ee''=X''/|X''|$,
  \begin{align*}
    \E\big[ |X|^s \big] & \geq \E\big[ |LX'|^s + |RX''|^s \big] \\
    & = \E\Big[ |X'|^s\, \E\big[|L\ee'|^s\big|(X',X'')\big] 
    + |X''|^s\, \E\big[|R\ee''|^s\big|(X',X'')\big] \Big] \\
    & \geq \kappa^*_s  \E[|X|^s] .
  \end{align*}
  This is a contradiction, 
  showing that our assumption $\E[|X|^s]<\infty$ cannot be true.

  To prove formula \eqref{eq.bbstar},
  first observe that for any two unit vectors $\ee_1,\ee_2\in\sphere^{d-1}$, 
  there exists some $\Theta\in\so(\setR^d)$ with $\ee_2=\Theta\ee_1$.
  Thus, since $(L,R)$ are invariant under conjugation by $\Theta$,
  and $\Theta$ preserves the norm, it follows that
  \begin{align*}
    \E[|L\ee_1|^s] = \E[|\Theta L\Theta^{-1}\ee_2|^s] = \E[|L\ee_2|^s] ,
  \end{align*}
  and likewise for $R$.
\end{proof}

\subsection{Regularity of the stationary state}
Below, we show that --- under rather natural conditions ---
the stationary state $\mu_\infty$ possesses a density of a certain Sobolev-regularity.
\begin{lem}\label{lem.psitozero} 
  Under Assumptions 1--3,
  let $\Psi$ be defined by \eqref{eq.stationary}. 
  Assume 
  that $\prb\{L^T\ee=R^T\ee=0\}=0$ and $\prb\{L^T\ee\neq0\neq R^T\ee\}>0$ for every unit vector $\ee$.
  Then the function $\rho\mapsto\sup_{|\ee| =1}\Psi(\rho\ee)$ decreases monotonically to zero.
  Consequently, the stationary distribution is absolutely continuous.
\end{lem}
The example in Section \ref{sct.cross} does not meet the hypotheses of this lemma,
and, indeed, its stationary distribution is not absolutely continuous.
\begin{proof}
  Define the auxiliary functions $\Phi,\Theta_\lambda:\sphere^{d-1}\to\setR$ by
  \begin{align}
    \Theta_\lambda(\ee) := \Psi(\lambda\ee), \quad \Phi(\ee) := \lim_{\lambda\to\infty}\Theta_\lambda(\ee). 
  \end{align}
  Recall that $\Psi$ is a mixture of Gaussians, 
  thus the function $  \Theta_\lambda$ is real, non-negative and non-increasing in $\lambda\in\setR_+$ for every fixed $\ee$.
  Therefore, the limit $\Phi(\ee)\in[0,1]$ exists for every $\ee$.
  As the decreasing limit of the family $\Theta_\lambda$ of continuous and bounded functions, 
  $\Phi$ is upper semi-continuous,
  and therefore attains its supremum, say at $\ee_*$, i.e.
  \begin{align}
    \Phi(\ee_*) = \theta := \sup_{|\ee|=1} \Phi(\ee)
  \end{align}
  for some $\ee_*\in\sphere^{d-1}$.
  Observe that $\lim_{\lambda\to\infty}\Psi(\lambda L^T\ee_*)=\Phi(L^T\ee_*/|L^T\ee_*|)$ for all $(L,R)$ with $L^T\ee_*\neq0$;
  and similarly if $R^T\ee_*\neq0$. 
  Denoting by $\indi_A$ the indicator function of a set $A$,
  by dominated convergence one can write
  \[
  \begin{split}
    \theta & = \lim_{\lambda\to\infty}\Psi(\lambda\ee_*) 
    = \lim_{\lambda\to\infty}\E[\Psi(\lambda L^T\ee_*)\Psi(\lambda R^T\ee_*)] \\
    & = \E\bigg[\indi_{L^T\ee_*\neq0}\indi_{R^T\ee_*\neq0}
    \Phi\bigg(\frac{L^T\ee_*}{|L^T\ee_*|}\bigg)\Phi\bigg(\frac{R^T\ee_*}{|R^T\ee_*|}\bigg)\bigg] \\
    & \qquad + \E\bigg[\indi_{L^T\ee_*\neq0}\indi_{L^T\ee_*=0}   \Phi\bigg(\frac{L^T\ee_*}{|L^T\ee_*|}\bigg) \bigg]
    + \E\bigg[\indi_{L^T\ee_*=0}\indi_{R^T\ee_*\neq0} \Phi\bigg(\frac{R^T\ee_*}{|R^T\ee_*|}\bigg) \bigg] \\
    & \leq q\theta^2 + (1-q) \theta, \\
  \end{split}
  \]
  with $q:=\prb\{L^T\ee_*\neq0\neq R^T\ee_*\}$, which is positive by assumption.
  Hence, either $\theta=0$ or $\theta=1$ (which we need to exclude).
  However,
  \begin{align}
    1 = \Phi(\ee_*) = \Psi(\ee_*) = \E[\exp(-\frac12\,\ee_*^TS\ee_*)]
  \end{align}
  if and only if $\ee_*^TS\ee_*=0$ a.s.
  But this means that $0=\ee_*^T\E[S]\ee = \ee_*^T\Sigma_*\ee_*$,
  contradicting the hypothesis that $\Sigma_*$ is positive definite in Assumption \ref{ass.2}.

  To conclude the proof,
  define $\ee_n\in\sphere$ as one maximizer for $\Theta_{\lambda_n}$ for, say $\lambda_n=n$,
  i.e.,
  \begin{align}
    \Theta_n(\ee_n) = \max_{\ee\in\sphere}\Theta_n(\ee).
  \end{align}
  By compactness of $\sphere$, we may assume w.l.o.g. that $\ee_n\to\ee'\in\sphere$ as $n\to\infty$.
  Now, for given $\epsilon>0$,
  there exists an $N=N(\epsilon)$ such that. $\Theta_N(\ee')<\epsilon/2$,
  and there exists a $\delta>0$ such that $|\Theta_N(\ee)-\Theta_N(\ee')|<\epsilon/2$ whenever $|\ee-\ee'|<\delta$.
  Since $\Theta_n$ converges monotonically in $n$,
  it follows that $\Theta_n(\ee)<\epsilon$ for all $n\geq N$ and $|\ee-\ee'|<\delta$,
  and thus also
  \begin{align}
    \Theta_n(\ee_n)<\epsilon
  \end{align}
  for all $n$ sufficiently large.
\end{proof}
The following is our main result on the regularity of $\mu_\infty$.
\begin{thm}
  \label{thm.smooth}
  Let Assumptions 1--3 hold. 
  Assume further that there exists a $\delta>0$ such that
  \begin{align}
    \label{eq.smoothcon1}
    m:=\sup_{|\ee| =1}\E\big[\min(|L^T\ee|,|R^T\ee|)^{-\delta}\big] < \infty,
  \end{align}
  and that there exists an $\overline a\geq\delta$ such that
  \begin{align}
    \label{eq.smoothcon2}
    M:=\sup_{|\ee| =1}\E\big[\max(|L^T\ee|,|R^T\ee|)^{-\overline a}\big] < \infty.
  \end{align}
  Then, for every $a\in(0,\overline a)$,
  there exists a finite constant $C(a)$ such that
  \begin{align}
    \label{eq.decaywitha}
    |\Psi(\xi)| \leq C(a) |\xi|^{-a}
  \end{align}
  holds for all $\xi\neq0$.
  In particular, if  $a>d/2$, then the density of $\mu_\infty$ 
  belongs to the Sobolev space $H^s(\setR^d)$ for any $s<a-d/2$.
\end{thm}
\begin{proof}
  Without loss of generality, we assume $m\geq1$ and $M\geq1$.
  Let $a\in(0,\overline a)$ be fixed.
  Because of condition \eqref{eq.smoothcon1}, Lemma \ref{lem.psitozero} applies.
  Thus, we find a $\rho>0$ such that
  \begin{align*}
    |\Psi(\xi)|\leq \frac1{3M} \quad \text{for all $|\xi|\geq\rho$,}
  \end{align*}
  with $M$ defined in \eqref{eq.smoothcon2}.
  Next, choose $\vartheta>0$ such that $m\vartheta^{\delta/2}\leq1/(3M)$.
  By H\"older's inequality, and in view of \eqref{eq.smoothcon1},
  this guarantees that
  \begin{equation}
    \label{eq.mixedintegrals}
    \begin{split}
      &\E\big[|\max(|L^T\ee|,|R^T\ee|)^{-a}\indi_{\min( |L^T\ee|,|R^T\ee| )<\vartheta}\big]  \\
      & \qquad \leq \big(\E\big[|\max(|L^T\ee|,|R^T\ee|)^{-\overline a}\big]\big)^{a/\overline a}
      \prb\{\min( |L^T\ee|,|R^T\ee| )<\vartheta\}^{1-a/\overline a}
      \leq Mm\vartheta^{\delta/2}\leq \frac13
    \end{split}
  \end{equation}
  for every unit vector $\ee$.
  Finally, let $C(a):=3(\rho/\vartheta)^a$.
  
  Now, fix a vector $\xi_0\neq0$ in $\setR^d$,
  and let $(L_k,R_k)$ for $k=1,2,\ldots$ be i.i.d. copies of $(L,R)$.
  Accordingly, the $\sigma$-algebra $\sigalg_k$ is the one generated by $(L_\ell,R_\ell)$ with $1\leq\ell\leq k$.
  Define random elements $\xi_k,\eta_k\in\setR^d$ inductively as follows.
  Given $\xi_{k-1}$ with $k\geq1$, 
  consider the two random vectors $v=L_k^T\xi_{k-1}$ and $w=R_k^T\xi_{k-1}$,
  and define
  \begin{align*}
    \xi_k = \indi_{|v|\geq|w|}v + \indi_{|v|<|w|}w, \quad
    \eta_k = \indi_{|v|<|w|}v + \indi_{|v|\geq|w|}w.
  \end{align*}
  Thus, $\xi_k$ is always the ``larger one'' of $v,w$.
  Further, introduce the random variable $\chi_k$ by
  \begin{align*}
    \chi_k=1/(3M)+\indi_{|\eta_k|<\vartheta|\xi_{k-1}|}.
  \end{align*}
  Two properties of $\chi_k$ will be important in the following:
  \begin{align}
    \label{eq.chiprop}
    \E\big[\chi_k\big|\sigalg_{k-1}\big]\leq\frac2{3M}
    \quad \text{and} \quad
    \E\big[|\xi_k|^{-a}\chi_k\big|\sigalg_{k-1}\big]\leq\frac23 |\xi_{k-1}|^{-a}.
  \end{align}
  The first part of \eqref{eq.chiprop} follows directly from the definitions of $\chi_k$ and $\vartheta$, 
  since
  \begin{align*}
    \E\big[\chi_k\big|\sigalg_{k-1}\big]
    &= \frac1{3M} + \prb\big(|\eta_k|/|\xi_{k-1}|<\vartheta\big|\sigalg_{k-1}\big) \\
    &\leq \frac1{3M} + \sup_{|\ee|=1}\prb\big( \min( |L_k^T\ee|,|R_k^T\ee| )<\vartheta \big)
    \leq \frac1{3M} + m\vartheta^\delta \leq \frac2{3M}.
  \end{align*}
  The second part of \eqref{eq.chiprop} follows with the help of \eqref{eq.mixedintegrals},
  since
  \begin{align*}
    \E\big[|\xi_k|^{-a}\chi_k\big|\sigalg_{k-1}\big]
    & \leq \E\big[ ( |\xi_k|/|\xi_{k-1}| )^{-a}\chi_k \big|\sigalg_{k-1} \big]|\xi_{k-1}|^{-a} \\
    & \leq \sup_{|\ee|=1}\E\big[ \max( |L^T\ee|,|R^T\ee| )^{-a} \big(1/(3M)+\indy_{\min( |L^T\ee|,|R^T\ee| )<\vartheta}\big)\big]|\xi_{k-1}|^{-a} \\
    & \leq \left(\frac{M}{3M} + \frac13\right)|\xi_{k-1}|^{-a} .
  \end{align*}
  As an intermediate step,
  we are going to verify that
  \begin{align}
    \label{eq.smoothrecurse}
    |\Psi(\xi_{k-1})| \leq \E\big[|\Psi(\xi_k)\chi_k\big|\sigalg_{k-1}\big] + \frac{C(a)}3 |\xi_{k-1}|^{-a}
  \end{align}
  holds almost surely.
  Observe that inequality \eqref{eq.smoothrecurse} is trivially true 
  if $|\xi_{k-1}|<\rho/\vartheta$ by our definition of $C(a)$,
  thus we may assume $|\xi_{k-1}|\geq\rho/\vartheta$.
  Since $\Psi$ satisfies \eqref{eq.stationary},
  \begin{align*}
    |\Psi(\xi_{k-1})| 
    &= \big| \E\big[\Psi(L_k^T\xi_{k-1})\Psi(R_k^T\xi_{k-1})\big|\sigalg_{k-1}\big] \big|
    \leq \E\big[ |\Psi(\xi_k)| |\Psi(\eta_k)| \big| \sigalg_{k-1} \big] \\
    & \leq \E\big[ |\Psi(\xi_k)| \big( (1/3M) \indi_{|\eta_k|\geq\rho} + \indi_{|\eta_k|<\rho} \big) \big| \sigalg_{k-1} \big]
    \leq \E\big[|\Psi(\xi_k)\chi_k\big|\sigalg_{k-1}\big].
  \end{align*}
  This shows \eqref{eq.smoothrecurse}.
  From here, we derive by induction that
  \begin{align}
    \label{eq.smoothinduce}
    |\Psi(\xi_0)| \leq \E\big[ |\Psi(\xi_\ell)|\chi_\ell\chi_{\ell-1}\cdots\chi_1 \big] + C(a)\big(1-(2/3)^\ell\big)|\xi_0|^{-a}.
  \end{align}
  Indeed, for $\ell=1$, the claim \eqref{eq.smoothinduce} is identical to \eqref{eq.smoothrecurse} with $k=1$.
  Now assume that \eqref{eq.smoothinduce} holds for some $\ell\geq1$.
  Inserting \eqref{eq.smoothrecurse} with $k=\ell+1$ into \eqref{eq.smoothinduce} gives
  \begin{align*}
    |\Psi(\xi_0)| 
    &\leq \E\Big[ \big|\E\big[|\Psi(\xi_{\ell+1})|\chi_{\ell+1}\big|\sigalg_\ell\big]+\frac{C(a)}3 |\xi_\ell|^{-a} \big|\chi_\ell\chi_{\ell-1}\cdots\chi_1 \Big] 
    + C(a)\big(1-(2/3)^\ell\big)|\xi_0|^{-a} \\
    &\leq \E\big[|\Psi(\xi_{\ell+1})|\chi_{\ell+1}\chi_\ell\cdots\chi_1\big] + \frac{C(a)}3 \E\big[|\xi_\ell|^{-a}\chi_\ell\chi_{\ell-1}\cdots\chi_1\big] 
    + C(a)\big(1-(2/3)^\ell\big)|\xi_0|^{-a}.
  \end{align*}
  Apply the second estimate from \eqref{eq.chiprop} to
  the second expectation above,
  \begin{align*}
    \E\big[|\xi_\ell|^{-a}\chi_\ell\chi_{\ell-1}\cdots\chi_1\big] 
    & = \E\Big[\E\big[|\xi_\ell|^{-a}\chi_\ell\big|\sigalg_{\ell-1}\big]\chi_{\ell-1}\cdots\chi_1\Big]
    \leq \frac23 \E\big[|\xi_{\ell-1}|^{-a}\chi_{\ell-1}\cdots\chi_1\big] \\
    & \leq \cdots \leq (2/3)^\ell|\xi_0|^{-a}.
  \end{align*}
  This verifies \eqref{eq.smoothinduce} with $\ell+1$ in place of $\ell$.
  To finish the proof, observe that \eqref{eq.smoothinduce} for every $\ell\geq1$ implies \eqref{eq.decaywitha},
  since, by \eqref{eq.chiprop},
  \begin{align*}
    0 \leq \E\big[ |\Psi(\xi_\ell)|\chi_\ell\chi_{\ell-1}\cdots\chi_1 \big] 
    & \leq \E\Big[ \E\big[\chi_\ell\big|\sigalg_{\ell-1}\big]\chi_{\ell-1}\cdots\chi_1 \Big]
    \leq \frac2{3M} \E\big[ \chi_{\ell-1}\cdots\chi_1 \big] \\
    & \leq \cdots \leq \big(2/(3M)\big)^\ell,
  \end{align*}
  which tends to zero as $\ell$ tends to infinity.

  The statement about Sobolev regularity of the density function follows 
  since if the characteristic function $\Psi$ satisfies \eqref{eq.decaywitha}, 
  then $|\xi|^{s}\Psi(\xi)$ belongs to $L^2(\setR^d)$ for every $s<a-d/2$.
  In turn, the inverse Fourier transform of $\Psi$ is a density function in the space $H^s(\setR^d)$.
\end{proof}

\section{Convergence to equilibrium}
\label{sct.dynamic}

Having established the existence of a stationary distribution for \eqref{eq.b},
we now turn to study transient solutions and prove their convergence towards equilibrium.

%
Transient solutions $\mu(t)$ to the initial value problem \eqref{eq.b} are easily obtained,
for an arbitrary initial probability measure $\mu_0$ on $\setR^d$,
by the classical Wild construction \cite{Wild1951}:
define inductively probability measures $\mu_n$ by
\begin{align*}
  \mu_{n+1} := \frac1{n+1} \sum_{j=0}^n Q_+[\mu_k,\mu_{n-j}] \qquad (n \geq 0),
\end{align*}
or, equivalently --- recalling \eqref{eq.qfourier} --- in Fourier space,
\begin{equation}
  \label{eq.wildrec}
  \hat \mu_{n+1}(\xi):= \frac{1}{n+1} \sum_{j=0}^{n} \E[\hat \mu_j(L^T \xi)\hat \mu_{n-j}(R^T \xi)].
\end{equation}
Then the infinite convex combination
\begin{equation}
  \label{eq.wildsum}
  \mu(t) = \sum_{n=0}^\infty e^{-t}(1-e^{-t})^n\mu_n
\end{equation}
defines a weak transient solution to \eqref{eq.b}:
the curve $t\mapsto\mu(t)$ is weakly continuous with respect to $t\ge0$,
and satisfies 
\begin{align}
  \label{eq.weak}
  \frac{\dd}{\dd t} \int_{\setR^d} \phi(v)\mu(t;\dd v) 
  +  \int_{\setR^d} \phi(v)\mu(t;\dd v) = \int_{\setR^d} \phi(v')Q_+[\mu(t)](\dd v')
\end{align}
for arbitrary test functions $\phi\in C_b(\setR^d)$, and at every $t>0$.
%
%

\subsection{Qualitative results by probabilistic methods}
%
In this section, we apply probabilistic methods related to the central limit theorem
to prove the following qualitative result on convergence to equilibrium.
\begin{thm}
  \label{thm.clt}
  Assume that the initial condition has unit temperature, $\tmp[\mu_0]=1$.
  Then the distributions $\mu_n$ converge weakly towards
  the stationary state $\mu_\infty$ obtained in Section \ref{sct.stationary} above.
  Consequently, 
  the associated transient solution $\mu(t)$ converges weakly to $\mu_\infty$ as $t\to\infty$.
\end{thm}
Our key tool is a probabilistic representation of the Wild sum \eqref{eq.wildsum},
following \cite{BasLadMat,GabettaRegazziniCLT}. 
On a sufficiently large probability space $(\Omega,\CF,\prb)$,
let the following be given:
\begin{itemize}
\item a sequence $(X_n)_{n\in\setN}$ of i.i.d. random vectors  with distribution $\mu_0$;
\item a sequence $\big((L_n,R_n)\big)_{n\in\setN}$ of i.i.d. random matrices, distributed as $(L,R)$;
\item a sequence $(I_n)_{n\in\setN}$ of independent integer random variables,
  each $I_n$ being uniformly distributed on the indices $\{1,2,\ldots,n\}$;
\end{itemize}
We assume further that $(I_n)_{n\in\setN}$, $(L_n,R_n)_{n\in\setN}$ and $(X_n)_{n\in\setN}$  are stochastically independent.
Define a random array of weights $[\beta_{j,n}: j=1,\dots,n+1]_{n \geq 0}$ recursively:
Let $\beta_{1,0}:=\eins$, $(\beta_{1,1},\beta_{2,1}):=(L_1,R_1)$, and for any $n\geq2$,
\begin{equation}
  \label{recursion}
  (\beta_{1,n}, \ldots,\beta_{n+1,n})  
  := (\beta_{1,n-1},\ldots,\beta_{I_n-1,n-1},  \beta_{I_n,n-1}L_n,  \beta_{I_n,n-1}R_n, \beta_{I_n+1,n-1},\ldots, \beta_{n,n-1}).
\end{equation}
Finally, setting
\begin{align}
  \label{eq.defv}
  W_n:=\sum_{j=1}^{n+1} \beta_{j,n} X_j,
\end{align}
one obtains the following alternative representation of the measures $\mu_n$ from \eqref{eq.wildsum}.
\begin{prp}\label{Prop:probint} 
  For every $n\geq 1$, one has
  \begin{equation}
    \label{fcharW}
    \hat \mu_n(\xi)=\E[\exp(i \xi^T W_n)].
  \end{equation}
\end{prp}
\begin{proof}
  The proof is analogous to the proof of Proposition 1 in \cite{BasLadMat} for the scalar case,
  with the only difference that here one needs to be careful about the order of multiplication
  of the (non-commutative) matrices $L_k$ and $R_k$.
  We briefly repeat the argument as we shall refer to parts of it later.

  We prove \eqref{fcharW} by induction on $n \geq 1$. 
  For $n=0$ equation \eqref{fcharW} is clearly true. 
  For $n \geq 1$, let $J_n$ be the random index such that 
  \[
  \beta_{j,n}=\left\{\begin{array}{cl} 
      L_1 \beta_{j,n}^{\ell} & \text{for every $j\leq J_n$}, \\
      R_1 \beta_{j,n}^{r} & \text{for every $j=J_n+1,\dots,n+1$},
    \end{array}\right.
  \]
  with suitable random  matrices $\beta_{j,n}^\ell$ and $\beta_{j,n}^r$.
  It can be shown --- see Proposition 1 of \cite{BasLadMat} ---
  that
  \begin{itemize}
  \item[(i)] $J_n$ is uniformly distributed on $\{1,\dots, n\}$, and 
  \item[(ii)] $(\beta_{j,n}^\ell)_{j=1,\dots,J_n}$ and $(\beta_{j,n}^r)_{j=J_n+1,\dots,n+1}$
    are conditionally independent given $J_n$ with conditional distribution equal to
    the distribution of $(\beta_{j,J_n-1})_{j=1,\dots,J_n}$ and $(\beta_{j,n-J_n})_{j=1,\dots,n+1-J_n}$, 
    respectively.
  \end{itemize}
  This shows that  
  \begin{align*}
    W_n=L_1 \sum_{j=1}^{J_n} \beta_{j,n}^\ell X_j + R_1 \sum_{j=J_n+1}^{n+1} \beta_{j,n}^r X_j   
    \distreq LW_{J_n}' + RW_{n+1-J_n}''
  \end{align*}
  for two  sequences $(W_0',\ldots,W_n')$ and $(W_0'',\ldots,W_n'')$ of i.i.d. copies of $(W_0,\ldots,W_n)$
  independent from  $J_n$ and $(L,R)$. 
  Hence, using the induction hypothesis, 
  \begin{align*}
    \E[e^{i\xi^T W_n}]& =\E[e^{i\xi^T( LW_{J_n-1}' + RW_{n-J_n}'' ) }]
    =\E[\E[e^{i(L^T\xi)^T W_{J_n-1}'}e^{i(R^T\xi)^T W_{n-J_n}''}|J_n,L,R] ] \\
    & = \E[\hat\mu_{J_n-1}(L^T\xi)\hat\mu_{n-J_n}(R^T\xi) ] 
    = \frac1{n+1} \sum_{j=0}^{n-1} \E[\hat\mu_i(L^T\xi)\hat\mu_{n-i}(R^T\xi) ] =\hat\mu_n(\xi).
    \qedhere
  \end{align*}
\end{proof}
As a first consequence of the probabilistic representation we prove 
\begin{lem}
  Under Assumption 1, and if $\mu_0$ is centered, 
  then $\mu(t)$ is a centered probability measure for all $t\geq0$.
  If additionally $\mu_0$ has finite temperature $\tmp[\mu_0]<\infty$,
  then this is conserved, $\tmp[\mu(t)]=\tmp[\mu_0]$.
  If instead $\tmp[\mu_0]=+\infty$, then also $\tmp[\mu(t)]=+\infty$ for all $t\geq0$.
\end{lem}
\begin{proof}
  It is clear that it suffices to prove the statement for $\mu_n$. 
  First of all note that \eqref{eq.general} yields
  $E[\|L\|_2^2+\|R\|_2^2]<+\infty$ and hence $E[\|L\|_2+\|R\|_2]<+\infty$. As a consequence,
  we shall show that $\E[\|\sum_{j=1}^{n+1} \beta_{jn}\|_2]<+\infty$. It clearly suffices to show 
  that  $ \E[\|\beta_{jn}\|_2]<+\infty$ for every $j=1,\dots,n+1$. Now each 
  $\beta_{jn}$ is a random product a random number, say $\delta_j$, of matrices $L$ and $R$. 
  Since $(L_k,R_k)_k$ are independent and identically distributed
  and, for each $k$, $\law(L_k)=\law(R_k)$, using the  subadditivity of the norm and the fact that
  $\delta_j \leq n$ a.s., we have 
  \[
  \E[\|\beta_{jn}\|_2]\leq \E[\prod_{k=1}^{\delta_j} \|L_k\|_2]
  \leq \E[\prod_{k=1}^{\delta_j} \max(1,\|L_k\|_2)] \leq (\E[\max(1,\|L_1\|_2)])^n<+\infty.
  \]  
    Propagation of centering follows now from the linearity of the collision rules \eqref{eq.rules}.
  Indeed
  \[
  \E[W_n]=\E[\sum_{j=1}^{n+1} \beta_{jn} X_j]=\E[\sum_{j=1}^{n+1} \beta_{jn}]\E[X_1] = 0,
  \]
  since $(X_j)_{j \geq 1}$ and $(\beta_{jn})_{j=1,\dots,n+1}$ are independent and 
  $\E[X_1]=0$. 
  As for the temperature, 
  \[
  \E[W_n^TW_n]=\sum_{j=1}^{n+1}\E[X_j^T \beta_{jn}^T \beta_{jn} X_j]
  +\sum_{i,j= 1 \atop i\not=j}^{n+1} \E[X_j^T \beta_{jn}^T \beta_{in} X_i].
  \]
  Clearly, since the $X_j$ are i.i.d.\ centered random variables,
  and are also independent of the $\beta$'s,
  \[
  \sum_{i,j= 1 \atop i\not=j}^{n+1} \E[X_j^T \beta_{jn}^T \beta_{in} X_i]=0,
  \quad \text{and} \quad
  \sum_{j=1}^{n+1}\E[X_j^T \beta_{jn}^T \beta_{jn} X_j]
  = \E[X_1^T \E[\sum_{j=1}^{n+1} \beta_{jn}^T \beta_{jn}]  X_1 ].
  \]
It remains to prove that 
  \begin{align}
    \label{eq.sumtoone}
   \E[\sum_{j=1}^{n+1} \beta_{jn}^T \beta_{jn}]=\eins.
  \end{align}
We shall prove \eqref{eq.sumtoone} by induction. 
 Set \(  \theta_{n}:=\sum_{j=1}^{n+1} \beta_{jn}^T \beta_{jn} \). 
We have $\E[\theta_n]=\eins$ for $n=0$ and $n=1$. For $n \geq 2$, 
with the same notation of the proof of Proposition \ref{Prop:probint},
we can write
\[
\begin{split}
\E[\theta_n]&= \E[\sum_{j=1}^{J_n}(L_1\beta_{jn}^{(l)})^T L_1\beta_{jn}^{(l)}
+\sum_{j=J_n+1}^{n+1}(R_1\beta_{jn}^{(r)})^T R_1\beta_{jn}^{(r)}]
\\
&=\E[\sum_{j=1}^{J_n} (\beta_{jn}^{(l)})^T L_1^T L_1 \beta_{jn}^{(l)}
+\sum_{j=J_n+1}^{n+1}(\beta_{jn}^{(r)})^T R_1^T R_1\beta_{jn}^{(r)}].
\\
\end{split}
\]
Since $\E[R_1^T R_1]=\E[L_1^T L_1]=\frac{1}{2}\eins$ and 
$(L_1,R_1)$ and $\beta_{jn}^{(l)},\beta_{jn}^{(r)}$ are stochastically independent,
recalling also $(i)$ and $(ii)$ of the proof of Proposition \ref{Prop:probint},
we can write
\[
\E[\theta_n]  =\frac{1}{2} \E[\sum_{j=1}^{J_n} (\beta_{jn}^{(l)})^T  \beta_{jn}^{(l)}
+\sum_{j=J_n+1}^{n+1}(\beta_{jn}^{(r)})^T \beta_{jn}^{(r)}] 
= \frac{1}{2n} \sum_{k=1}^n \E[\theta_{k-1}+\theta_{n-k}].
\]
Hence we can conclude by  the induction hypothesis. 
\end{proof}

Theorem \ref{thm.clt} will be proven by studying the distributional convergence
of the random variables $W_n$ with methods related to the central limit theorem.
For definiteness, denote by $\sigalg$ the $\sigma$-field 
generated by the weights $\beta_{n,j}$, for $n\geq0$ with $1\leq j\leq n+1$, respectively.
Accordingly, introduce the conditional covariance matrices
\begin{align}
  \label{defnun}
  S_n:=\E\big[ W_nW_n^T \big|\sigalg \big] = \sum_{j=1}^{n+1}\beta_{n,j}S_0\beta_{n,j}^T
  \quad \text{and} \quad
  \nu_n :=\law(S_n),
\end{align}
so that in particular $S_0=\cov(\mu_0)$ a.s., and $\nu_0:=\delta_{S_0}$.
In a series of lemmas, we show the weak convergence of the $\nu_n$
towards $\nu_\infty$ obtained in Theorem \ref{thm.exuniq}.
First, we derive an a priori estimate on the weights.
\begin{lem}
  \label{Lemma.eq} 
  Under Assumptions \ref{ass.1} and \ref{ass.3new},
  there exists a constant $C$ such that
  \begin{align}
    \label{eq.mest}
    M^{(n)}_p := \sum_{j=1}^{n+1} \sup_{|\ee|=1}\E[|\beta_{jn}^T\ee|^p]  
    \leq Cn^{ \kappa_p-1}
  \end{align}  
  for every $n\geq1$.
  In particular, $M^{(n)}_p \to0$ as $n\to\infty$.
\end{lem}
\begin{proof}
  We first show a similar estimate for a related quantity, namely
  \begin{align}
    \label{eq.mest2}
    M^{(n)}_{\weight}
    := \sum_{j=1}^{n+1} \sup_{ \eta\not=0}\E\Big[\frac{\weight(\beta_{jn}^T\eta)}{\weight(\eta)}\Big]  
    \leq \prod_{m=1}^{n}\Big(1-\frac{1- \kappa_p}{m}\Big) 
    =\frac{\Gamma(\kappa_p+n)}{n!\Gamma(\kappa_p)}.
  \end{align}
  Here $\Gamma$ denotes Euler's function.
  We proceed inductively with respect to $n$.
  By construction of the weights,
  we have
  \begin{align}
    \label{eq.recurse2}
    M^{(n)}_{\weight} = \frac1n\sum_{k=1}^n\bigg( 
    \sum_{j=1 \atop j\neq k}^n \sup_{\eta \neq 0}\E[ \frac{\weight(\beta_{j,{n-1}}^T\eta)}{{\weight}(\eta)}] 
    + \sup_{\eta \neq 0} \E[  \frac{\weight(L_n^T \beta_{k,{n-1}}^T \eta)}{\weight(\eta)}] + 
    \sup_{ \eta \neq  0} \E[ \frac{\weight(R_n^T \beta_{k,{n-1}}^T \eta)}{\weight(\eta)}] \bigg) .
  \end{align}
  Since
  $(L_n,R_n)$ and $\beta_{k,,{n-1}}$ are independent, 
  it follows that
  \begin{align*}
    \E[  \frac{\weight(L_n^T \beta_{k,{n-1}}^T \eta)}{\weight(\eta)}]
    & = \E\Big[\E\big[\frac{\weight(L_n^T \beta_{k,{n-1}}^T \eta)}{\weight(\eta)}   \big|\sigalg_n\big]\Big] \\
    & \leq \E\Big [ \frac{\weight(\beta_{k,{n-1}}^T \eta)}{\weight(\eta)}  \Big] 
    \sup_{ \eta' \neq 0} \E \Big[ \frac{ \weight(L_n^T \eta')}{\weight(\eta')}\Big ]
    \leq \kappa_p\E\Big [ \frac{\weight(\beta_{k,{n-1}}^T \eta)}{\weight(\eta)}  \Big],
  \end{align*}
  where Assumption \ref{ass.3new} and the symmetry $(L,R)\distreq(R,L)$ have been used in the last step.
  Thus, in  \eqref{eq.recurse2} we have
  \[
  \sup_{\eta \neq 0}\E \Big[ \frac{{\weight}(L_n^T \beta_{k,{n-1}}^T \eta)}{\weight(\eta)}\Big] 
  +  \sup_{\eta \neq 0}\E\Big[ \frac{{\weight}(R_n^T \beta_{k,{n-1}}^T \eta')}{\weight(\eta')}\Big] 
  \leq  \kappa_p \sup_{\eta\not=0} \E\Big[ \frac{\weight(\beta_{k,{n-1}}^T \eta)}{\weight(\eta)}  \Big]
  \]
  which leads to the recursive relation
  \begin{align*}
    M^{(n)}_\weight \leq \bigg( 1-\frac{1-\kappa_p}{n} \bigg) M^{(n-1)}_\weight
  \end{align*}
  and implies \eqref{eq.mest2}.
  The original claim \eqref{eq.mest} now follows immediately, 
  since $M^{(n)}_p \leq \bar \weight M^{(n)}_\weight$ by \eqref{eq.boundass}.
\end{proof}
\begin{lem}
  \label{lemma_Sn}
  Under Assumptions 1--3, if $\tmp[\mu_0]=1$,
  and with $\nu_\infty$ be defined in Theorem \ref{thm.exuniq},
  it follows that
  \begin{equation}
    \label{recur1}
    \fnorm_\wweight( \nu_n,\nu_{\infty}) \leq \frac{\lambda}{n} \sum_{i=1}^{n} 
    \fnorm_\wweight(\nu_{i-1},\nu_{\infty})
  \end{equation}
  where $\fnorm_\wweight$ is the metric introduced in \eqref{eq.ourmetric},
  with $\alpha$ large enough such that $\lambda<1$ in \eqref{eq.deflambda}. 
  Consequently, the $\nu_n$ converge weakly to $\nu_\infty$.
\end{lem}
\begin{proof}
  Since $\nu_\infty$ is a solution to \eqref{eq.aux},
  we can write
  \begin{equation}
    \label{F4}
    \nu_\infty = T[\nu_\infty,\nu_\infty] = \frac{1}{n} \sum_{i=1}^n T[\nu_\infty,\nu_\infty]. 
  \end{equation}
  Recall the definition of $\nu_n$ in \eqref{defnun}.
  We prove that for every $n \geq 1$,
  \begin{equation}
    \label{F3}
    \nu_{n} = \frac{1}{n} \sum_{i=1}^n T[\nu_{i-1},\nu_{n-i}]. 
  \end{equation}
  Indeed, with $\beta_{j,n}^{\ell}$, $\beta_{j,n}^{r}$ and $J$ being defined as in the proof of Proposition \ref{Prop:probint},
  we have
  \[
  S_{n}=L_1 \Big ( \sum_{j=1}^{J}\beta_{j,n}^\ell S_0(\beta_{j,n}^\ell)^T \Big ) L_1^T
  + R_1 \Big ( \sum_{j=J+1}^{n+1} \beta_{j,n}^r S_0 (\beta_{j,n}^r)^T \Big ) R_1^T .
  \]
  Now \eqref{F3} follows since the expressions in the round brackets 
  are independent of each other and of $(L_1,R_1)$,
  and are distributed like $S_{J-1}$ and $S_{n-J}$, respectively.
  Using the representations \eqref{F4}\&\eqref{F3} and the convexity of $\fnorm_\wweight$,
  see \eqref{eq.fnormconvex},
  we obtain
  \begin{align*}
    \fnorm_\wweight(\nu_{n},\nu_\infty ) \leq \frac{1}{n} \sum_{i=1}^n \fnorm_\wweight (T[\nu_{i-1},\nu_{n-i}],T[\nu_\infty ,\nu_\infty ]).
  \end{align*}
  Now \eqref{recur1} follows from \eqref{eq.fwcontract}, 
  with $\lambda<1$ defined as in \eqref{eq.deflambda}. 
  To conclude the proof, 
  we recall that if a sequence $(a_n)_n$ of non-negative numbers
  satisfies the inequalities
  \[
  a_{n+1} \leq \frac{\lambda}{n} \sum_{i=1}^n a_i 
  \]
  for every $n \geq 1$,
  with some fixed constant $\lambda \in (0,1)$, 
  then 
  \[a_n \leq a_1 \prod_{j=1}^{n-1}(1-\frac{1-\lambda}{j}) \leq Ca_1n^{-(1-\lambda)}, 
  \]
  which tends to zero as $n\to\infty$.
  Inequality \eqref{recur1} thus implies $\fnorm_\wweight(\nu_{n},\nu_{\infty})\to0$,
  proving weak convergence of $\nu_n$ to $\nu_{\infty}$.
\end{proof}
Weak convergence of the $\nu_n$ is not quite enough to apply the methods
of the central limit theorem.
We also need a uniform control on the second momenta,
which is related to the Lindeberg condition.
\begin{lem}
  \label{lemmalindeberg} 
  Under Assumption 1--3, and if $\tmp[\mu_0]=1$,
  it follows that for every $\eps>0$ and every $\xi\in\setR^d$
  \begin{align}
    \label{cond2-bis}
    \lim_{n \to +\infty} \E\Big [\sum_{j=1}^{n+1} |\xi^T\beta_{j n} X_j|^2 \J\{|\xi^T\beta_{j n}X_j|> \eps \}\Big ] = 0.
  \end{align}
\end{lem}
\begin{proof}
  Let $\xi\in\setR^d$ be fixed.
  For notational simplicity set
  \begin{align}
    \label{eq.defxjn}
    X_{jn}:=\xi^T(\beta_{j n} X_{j})=(\beta_{j n}^T \xi)^T X_{j}. 
  \end{align}
  Since $\E[|X_j|^2]=d$ by hypothesis, 
  for every $\eta>0$ there is a $K=K(\eta)$ such that $\E[|X_j|^2\J\{|X_j|>K\}]\leq \eta$. 
  At this stage observe that, by \eqref{eq.sumtoone},
  \begin{align}
    \label{eq.sumtoone2}
    \E\Big  [\sum_{j=1}^{n+1}  \|\beta_{j n}\|_2^2\Big ] 
    =\E\Big [\sum_{j=1}^{n+1} \tr(\beta_{j n}^T \beta_{j n}) \Big ]
    = \tr \eins =d.
  \end{align}
  Moreover, the rule \eqref{eq.submult2} implies
  \begin{align}
    \label{eq.xjnest}
    |X_{jn}|^2 \leq \|\beta_{j n}\|_2^2 |\xi|^2|X_j|^2 \quad \text{and} \quad |X_{jn}|^2 \leq |\beta_{j n}^T\xi|^2 |X_j|^2,
  \end{align}
  and so we obtain
  \begin{align*}
    \E\Big [\sum_{j=1}^{n+1} |X_{jn}|^2 \J\{|X_{jn}|> \eps \}\Big ] 
    & \leq \E\Big [\sum_{j=1}^{n+1} \|\beta_{j n}\|_2^2\Big ] |\xi|^2 \E\Big [|X_1|^2\J\{|X_1|>K\}\Big ] \\
    & \qquad  +\E\Big [\sum_{j=1}^{n+1} |X_{jn}|^2 \J\{|X_{jn}|> \eps,|X_j|\leq K \}\Big ] \\
    & \leq d \eta|\xi|^2 + K^2  \E\Big [\sum_{j=1}^{n+1} |\beta_{j n}^T \xi|^2 \J\{|\beta_{j n}^T \xi|> \frac{\eps}{K} \}\Big ].
  \end{align*}
  Let $q:=p/2$ and $q':=p/(p-2)$ be its conjugate exponent.
  Then, by H\"older's and by Markov's inequality,
  \begin{align*}
    \E\Big [ |\beta_{j n}^T \xi|^2 \J\Big\{|\beta_{j n}^T \xi|> \frac{\eps}{K} \Big\}\Big ]
    & \leq \left( \E \big [ |\beta_{j n}^T \xi|^p \big]\right)^{\frac1q} 
    \left(\prb\Big\{|\beta_{j n}^T \xi|> \frac{\eps}{K}\Big \}\right)^{\frac{1}{q'}} \\
    &\leq \Big( \frac{K}{\eps} \Big)^{\frac{p}{q'}} \E\Big[ |\beta_{j n}^T \xi|^p\Big]^{\frac{1}{q'}}
    \left( \E \big [ |\beta_{j n}^T \xi|^p \big]\right)^{\frac1q} .
  \end{align*}
  Hence
  \begin{align*}
    \E\Big [\sum_{j=1}^{n+1} |X_{jn}|^2 \J\{|X_{jn}|> \eps \}\Big ] & \leq 
    \eta d |\xi|^2 + K^2 \Big( \frac{K}{\eps} \Big)^{\frac{p}{q'}}
    \sum_{j=1}^{n+1} \E\Big[ |\beta_{j n}^T \xi|^p\Big] \\
    & \leq \eta d |\xi|^2 + K^2 \Big( \frac{K}{\eps} \Big)^{\frac{p}{q'}} |\xi|^p
    \sum_{j=1}^{n+1} \sup_{\ee:|\ee|=1} \E[|\beta_{j n}^T\ee|^p] \\
    & = \eta d |\xi|^2 + K^2 \Big( \frac{K}{\eps} \Big)^{\frac{p}{q'}}|\xi|^p M_p^{(n)}.
  \end{align*}
  With \eqref{eq.mest}, 
  we finally get
  \[
  \limsup_{n \to +\infty} 
  \E\Big [\sum_{j=1}^{n+1}| X_{j n}|^2 \J\{|X_{jn}|> \eps \}\Big ]
  \leq  \eta d|\xi|^2.
  \]  
  Since $\eta$ is arbitrary, this shows \eqref{cond2-bis}.
\end{proof}
This concludes our preparations for the proof of Theorem \ref{thm.clt}.
\begin{proof}[Proof of Theorem \ref{thm.clt}]
  Let $S_\infty$ be a random variable with law $\nu_\infty$. 
  Let $(Y_j)_j$ be a sequence of i.i.d. random vectors 
  with Gaussian distribution of zero mean and covariance matrix $S_0$.
  Observe that $S_n=\cov[\sum_{j=1}^{n+1} \beta_{j n} Y_j|\sigalg]$. 
  Since a linear combination of Gaussian random variables is a Gaussian random variable,
  \begin{align*}
    \E\Big [\exp\{i  \xi^T \sum_{j=1}^{n+1} (\beta_{j n} Y_{j}) \}\Big]
    &= \E\Big [\E\Big [\exp\{i  \xi^T \sum_{j=1}^{n+1} (\beta_{j n} Y_{j}) \}\Big|\sigalg\Big]\Big] \\
    &=\E\Big[\exp\{ -\frac{1}{2}  (\xi^T S_n \xi)  \}\Big] 
    \to \E\Big[\exp\{ -\frac{1}{2}  (\xi^T S_\infty \xi)  \}\Big]
  \end{align*}
  Here we used the distributional convergence of $S_n$ to $S_\infty$ by Lemma \ref{lemma_Sn}.
  Thus, Theorem \ref{thm.clt} follows if we prove that
  \begin{equation}
    \label{resto}
    \Big|\E\Big[\exp\{i \xi^T \sum_{j=1}^{n+1} (\beta_{j n} X_{j}) \}\Big]
    -\E\Big[\exp\{i \xi^T \sum_{j=1}^{n+1} (\beta_{j n} Y_{j})\}\Big]\Big| \to 0.
  \end{equation}
  Define $X_{jn}$ as in \eqref{eq.defxjn} above, and accordingly $Y_{jn}:=\xi^T(\beta_{j n}Y_j)$,
  as well as the conditional second momenta in direction of $\xi$,
  \[
  \sigma_{jn}^2:=\E[ X_{jn}^2|\CB]=\E[Y_{jn}^2|\CB]=\xi^T [\beta_{j n} S_0 \beta_{j n}^T] \xi .
  \]
  Also, we introduce the characteristic functions
  \[ 
  \phi_{jn}:=\E[\exp \{ i X_{jn} \}|\sigalg], \qquad \psi_{jn}:=\E[\exp\{i Y_{jn }\}|\sigalg] .
  \]
  Clearly, \eqref{resto} is equivalent to
  \[
  \Big|\E\Big [\prod_j^{n+1} \phi_{jn}-\prod_j^{n+1} \psi_{jn}\Big]\Big| \to 0.
  \]
  Since for arbitrary complex numbers $z_1,\dots,z_n$ and $z_1',\dots,z_n'$ of modulus less or equal to one,
  \begin{equation}
    \label{fact1}
    |\prod_{j=1}^nz_j - \prod_{j=1}^n z_j'| \leq \sum_{j=1}^n|z_j- z_j'|,
  \end{equation}
  we can write
  \begin{equation}
    \label{firstboundch}
    \begin{split}
      \Big |\E \Big [ \prod_j \phi_{jn}-\prod_j \psi_{jn}\Big ]\Big |& 
      \leq \E\Big [\sum_j \Big |\phi_{jn}-\psi_{jn}\Big | \Big  ] \\
      & \leq \E\Big [\sum_j \Big |\phi_{jn}-1+\frac{1}{2} \sigma_{jn}^2\Big |\Big ]
      +\E\Big [\sum_j \Big |\psi_{jn}-1+\frac{1}{2} \sigma_{jn}^2\Big |\Big ]
    \end{split}
  \end{equation}
  By formula \eqref{eq.taylor2} from the appendix,
  any centered random variable $Z$ satisfies
  \begin{equation}
    \label{taylor2}
    \Big|\E[e^{iZ}]-1+\frac12 \E[Z^2]\Big| \leq \E\Big[\min\Big\{Z^2 , \frac16 |Z|^3 \Big\}\Big].
  \end{equation}
  Observing that
  \[
  \E[ Y_{j n} |\sigalg]=\E[X_{j n}|\sigalg] = 0,
  \]
  we can use \eqref{firstboundch}-\eqref{taylor2} and the definition of $\phi_{jn}$ and $\psi_{jn}$ 
  to conclude 
  \[
  \big |\E\big[\prod_j^{n+1} \phi_{jn}-\prod_j^{n+1} \psi_{jn}\big]\big| \leq A_n(\eps)+B_n(\eps)
  \]
  for arbitrary $\eps>0$,
  where
  \[
  A_n(\eps):=\frac{1}{6} \sum_j \E[|X_{jn}|^3 \J \{|X_{jn}|\leq \eps  \} ]
  +  \sum_j \E[|X_{jn}|^2 \J \{|X_{jn}|> \eps  \} ],
  \]
  and $B_n(\eps)$ is the respective expression for $Y_{jn}$.
  In view of \eqref{eq.xjnest}, we can estimate
  \[
  A_n 
  \leq  \frac{\eps}{6} |\xi|^2 \E|X_1|^2 \sum_j\E[\|\beta_{j n}\|_2^2] + \sum_j \E[ |X_{j n}|^2 \J\{|X_{j n}| > \eps \} ].
  \]
  The first term on the right-hand side can be made arbitrarily small 
  for an appropriate choice of $\eps>0$;
  recall that $\sum_j\E[\|\beta_{j n}\|_2^2] =d$ by \eqref{eq.sumtoone2}.
  And the second term converges to zero for $n\to\infty$ by \eqref{cond2-bis},
  independently of $\eps$.
  The analogous argument applies to $B_n(\eps)$.
\end{proof}

\subsection{Quantitative results in Fourier metrics}
Under a very weak additional hypothesis on the initial condition,
the previously obtained qualitative convergence result 
can even be quantified on the level of the Fourier transform.
\begin{thm}
  \label{thm.fourierconverge}
  Let Assumptions 1--3 hold, and let $\Sigma_*$ be positive definite.
  If the initial condition $\mu_0$ belongs to a class $\cls_{2+\sigma}(\setR^d)$ with some $\sigma>0$,
  i.e., if its absolute moment of order $2+\sigma$ is finite,
  then there exists a constant $\Lambda\in(0,1)$ that depends on $\mu_0$ only through $\sigma$,
  such that
  \begin{align}
    \label{eq.fourierconverge}
    \big|\hat\mu(t;\xi) - \hat\mu_\infty(\xi) \big|
    \leq C(\mu_0)\max\big(|\xi|^2,|\xi|^{2+\sigma}\big)e^{-(1-\Lambda) t}
  \end{align}
  holds for all $t\geq0$ and all $\xi\in\setR^d$.
\end{thm}

The appropriate Fourier distance in the current context is given by
\begin{align*}
  \newd_\weight(\mu_1,\mu_2) 
  & = \sup_{\xi\neq0} \big( \weight(\xi)^{-1}\big|\hat\mu_1(\xi)-\hat\mu_2(\xi)+\half\xi^T(\cov(\mu_1)-\cov(\mu_2))\xi\big| \big) \\
  & \qquad + \alpha \|\cov(\mu_1)-\cov(\mu_2)\|_\infty,
\end{align*}
with a parameter $\alpha>0$ that needs to be determined.
This definition is similar to the one for $\fnorm_\wweight$ given in \eqref{eq.ourmetric}.
\begin{lem}
  The distance $\newd_\weight$ defines a metric on $\cls_p(\setR^d)$ with the following convexity property:
  If $\mu_1=s\mu_1'+(1-s)\mu_1''$ and $\mu_2=s\mu_2'+(1-s)\mu_2''$ are two convex combinations 
  of measures in $\cls_p(\setR^d)$, with $s\in[0,1]$,
  then
  \begin{align}
    \label{eq.newdconvex}
    \newd_\weight(\mu_1,\mu_2) \leq s\newd_\weight(\mu_1',\mu_2') + (1-s)\newd_\weight(\mu_1'',\mu_2''). 
  \end{align}
  Moreover, $\newd_\weight$ provides a pointwise control on the Fourier transforms,
  \begin{align}
    \label{eq.newdcontrol}
    \big| \hat\mu_1(\xi)-\hat\mu_2(\xi) \big| \le K\max\big( |\xi|^2,|\xi|^p \big) \,\newd_\weight(\mu_1,\mu_2),
  \end{align}
  with a constant $K$ that only depends on $\alpha$ and $\weight$.
  Finally, convergence $\newd_\weight(\mu_k,\mu_*)\to0$ for a sequence $(\mu_k)$ in $\cls_p(\setR^d)$ 
  implies weak convergence of $\mu_k$ to $\mu_*$ and convergence of the covariance matrices.
\end{lem}
We omit the proof, 
which is very similar to (and easier than) the proof of Lemma \ref{lem.fnorm}.
\begin{lem}
  \label{lem.fcontract}
  Let Assumptions 1--3 hold. Let $\mu_1,\,\mu_2\in\cls_p(\setR^d)$.
  Then 
  \begin{align}
    \newd_\weight\big(Q_+[\mu_1],Q_+[\mu_2]\big) \leq \lambda \newd_\weight(\mu_1,\mu_2)
  \end{align}
  with the same constant $\lambda$ defined in \eqref{eq.deflambda}.
\end{lem}
The proof combines the elements 
from the proofs of Lemma \ref{Lemma4} and Lemma \ref{lem.fwcontract}.
We leave it to the reader to adapt the technical details accordingly.
\begin{proof}[Proof of Theorem \ref{thm.fourierconverge}]
  Without loss of generality, we may assume that $\sigma<1$ and $p':=2+\sigma<p$.
  By Lemma \ref{lem.prime}, Assumption \ref{ass.3new} is satisfied for $p'$ with a suitable weight function $\weight'$.
  Now apply Lemma \ref{lem.fcontract} with $\weight'$ instead of $\weight$.
  Choosing $\alpha>0$ sufficiently large in the respective metric $\newd_{\weight'}$,
  one obtains
  \begin{align*}
    \newd_{\weight'}\big( Q_+[\mu_1],Q_+[\mu_2] \big) \leq \Lambda \newd_{\weight'} \big(\mu_1,\mu_2 \big)
    \quad \text{with} \quad
    \Lambda := \max\big(\kappa_{p'},\kappa+2\alpha^{-1}d^{p/2}\kappa_{p'}\big) < 1.
  \end{align*}
  Combining the convexity property  \eqref{eq.newdconvex}
  with the fact that, for any $s\geq0$,
  \begin{align*}
    \newd_{\weight'}\big(Q_+[\mu(s)],\mu_\infty\big) 
    = \newd_{\weight'}\big( Q_+[\mu(s)],Q_+[\mu_\infty] \big) \leq \Lambda \newd_{\weight'}(\mu(s),\mu_\infty)
  \end{align*}
  it follows that
  \begin{align*}
    \newd_{\weight'}(\mu(t),\mu_\infty) 
    &\leq e^{-t}\newd_{\weight'}(\mu_0,\mu_\infty) + \int_0^t e^{-(t-s)}\newd_{\weight'}\big( Q_+[\mu(s)],\mu_\infty \big)\dd s \\
    &\leq e^{-t}\newd_{\weight'}(\mu_0,\mu_\infty) + \Lambda \int_0^t e^{-(t-s)}\newd_{\weight'}(\mu(s),\mu_\infty)\dd s .
  \end{align*}
  Gronwall's lemma now yields 
  \begin{align}
    \label{eq.expconvergence}
    \newd_{\weight'}(\mu(t),\mu_\infty) \leq \newd_{\weight'}(\mu_0,\mu_\infty)e^{-(1-\Lambda)t}.
  \end{align}
  From \eqref{eq.expconvergence}, the claim \eqref{eq.fourierconverge} follows with the help of \eqref{eq.newdcontrol}
\end{proof}

\section{Divergence for data of infinite energy}
\label{sct.explosion}
In this section, we prove that the solution to the initial value problem \eqref{eq.b} cannot converge
to a stationary state if the temperature of the initial condition is infinite.
For $z\in\setR^d$ and $r>0$, 
let $\ball_r(z):=\{x \in \setR^d: |x-z| \leq r\}$ be the closed ball of radius $r$ around $z$.
\begin{thm}
  \label{thm.explosion}
  Let Assumptions 1--3 hold, 
  and suppose that the stationary solution $\mu_\infty$ of unit temperature has no concentration in zero. 
  If $\mu(t)$ is the solution to \eqref{eq.b} with $\tmp[\mu_0]=+\infty$,
  then $\mu(t)[\ball_r(0)]$ converges to zero as $t\to\infty$, for every $r>0$. 
\end{thm}
In the following, let $\LL>0$ be a parameter,
which is large enough such that
\begin{align*}
  a_\LL := \mu_0[\ball_\LL(0)] > 0.
\end{align*}
Accordingly, we decompose $\mu_0=a_\LL\underline\mu_0+(1-a_\LL)\overline\mu_0$,
where the probability measures $\underline \mu_0$ and $\overline \mu_0$ are given by
\begin{align*}
  \underline \mu_0[\cdot]:=\frac{\mu_0[ \cdot \cap\ball_\LL(0)]}{a_\LL} 
  \quad \text{and} \quad
  \overline\mu_0[\cdot]:=\frac{\mu_0[ \cdot \cap \ball_\LL(0)^C]}{1-a_\LL},
\end{align*}
and introduce expectation and variance of the restriction,
\begin{align*}
  m_\LL:=\int_{\setR^d} v \underline \mu_0(\dd v) 
  \quad \text{and} \quad  
  \sigma_\LL^2:=\int_{\setR^d} (v-m_\LL)^2 \underline \mu_0(\dd v).
\end{align*}
Finally, recall that $\Psi=\hat \mu_\infty$ is the characteristic function 
of the stationary solution with covariance $\Sigma_*$ ---
see \eqref{eq.stationary} ---
and define
\[
H(u ):=\int_{\setR^d} e^{-\frac{|\xi|^2}{2}} \Psi( \xi u ) \dd\xi
=\frac{1}{u^d} \int_{\setR^d} e^{-\frac{|\xi|^2}{2u^2}} \Psi( \xi ) \dd\xi .
\]
Observe that $H(u)\to0$ as $u\to\infty$ if and only if $\mu_\infty$ has no concentration in $0$. 
The key is to establish the following time-dependent upper bound on $\mu(t)[\ball_\RR(z)]$.
\begin{prp}
  \label{prp.explosion}
  Under Assumptions 1--3, 
  and if $\RR>0$ and $ \sigma_\LL>\sqrt{2}\RR$, 
  then 
  \begin{align}
    \label{eq.explosion}
    \sup_{z\in\setR^d}\mu(t)[\ball_\RR(z)]
    \leq C \Big [ (\LL/\RR)^p \max(e^{-(1-\Lambda) t},e^{-\frac{2}{p}(1-\kappa_p)\cdot t}) + H(\sqrt{a}\sigma_\LL/\RR\sqrt{d}) \Big ]
  \end{align}
  where $\Lambda\in(0,1)$ is the constant defined in Theorem \ref{thm.fourierconverge} for $2+\sigma=p$.
\end{prp}

\begin{proof}
  For the sake of notational simplicity write $a$ for $a_\LL$ and $m$ for $m_\LL$. 
  Define the array $\{\beta_{j,n}\}_{n \geq 0,1\leq j\leq n+1}$ of random weights as in \eqref{recursion}.
  In addition,
  consider a sequence of independent and identically distributed random elements
  \[
  (\alpha_j,\xu_j,\xo_j)_{j \geq 1}
  \]
  where $\alpha_j,\xu_j,\xo_j$ are independent of each other, and independent of the $\beta_{j,n}$.
  The distributions of $\alpha_j$, $\xu_j$ and $\xo_j$, respectively, 
  are given by
  \begin{align*}
    \prb\{\alpha_j=1\}=1-\prb\{\alpha_j=0\}=a, \quad  
    \underline\mu_0 = \law \xu_j, \quad \overline\mu_0 = \law\xo_j.
  \end{align*}
  It is easily seen that $\xx_j:=\alpha_j\xu_j+(1-\alpha_j)\xo_j$ has law $\mu_0$, 
  that $|\xu_j|\leq \LL \leq |\xo_j|$ a.s.,
  and that $m=\E [\xu_j]$ and $\sigma_\LL^2=\E[|\xu_j-m|^2]$, for all $j\geq1$.
  In particular, in the Wild representation \eqref{eq.wildsum} of the transient solution $\mu(t)$,
  we can write $\mu_n =\law[\sum_{j=1}^{n+1} \beta_{j,n} X_j]$.

  Further, define the Gaussian density $\gamma(y)=\exp(-|y|^2/2)$.
  By Parseval's identity,
  we obtain for any $\RR>0$ and $z$ in $\setR^d$
  that
  \[
  \int_{\setR^d} \gamma\left(\frac{v-z}{\RR}\right) \mu(t;\dd v) 
  = \frac{\RR^d}{(2\pi)^{\frac{d}{2}}} \int_{\setR^d}  e^{-i \xi^T z} \gamma(\RR\xi) \hat \mu(t;\xi) \dd\xi.
  \]
  Since $\gamma( (v-z)/\RR)\geq e^{-1/2}$ for all $v\in\ball_\RR(z)$,
  it follows further that
  \begin{align}
    \label{eq.defIt}
    \frac{(2\pi)^{\frac{d}{2}}}{e^{\frac{1}{2}}} \sup_{z \in \setR^d}\mu(t)[\ball_\RR(z)]
    \leq (2\pi)^{\frac{d}{2}}\int_{\setR^d} \gamma\left(\frac{v-z}{\RR}\right) \mu(t;\dd v) 
   \leq \int_{\setR^d} \gamma(\xi)|\hat\mu(t;\xi/\RR)| \dd\xi =:I(t).
  \end{align}
  Denote by $\phi(\xi)$ the characteristic function of $\xu_1-m$ 
  and by $\Phi$ the characteristic function of $\xo_1$. 
  For the sake of simplicity set $\tau_n=e^{-t}(1-e^{-t})^n$.
  The Wild sum representation implies
  \[
  |\hat\mu(t;\xi)| 
  = \Big| \sum_{n \geq 0} \tau_n
    \E\Big [\prod_{j=1}^{n+1}\phi(\alpha_j \beta_{j {n}}^T\xi) e^{i\alpha_j \beta_{j {n}}m} 
    \prod_{j=1}^{n+1}\Phi\big((1-\alpha_j) \beta_{j {n}}^T\xi\big)\Big ] \Big|
  \leq \sum_{n \geq 0} \tau_n \E\Big[\Big|\prod_{j=1}^{n+1} \phi(\alpha_j \beta_{j {n}}^T\xi) \Big| \Big] .
  \]
  For further estimation,
  introduce the semi-definite matrix $\sqrt{\Sigma_\LL}$ as the square root of $\Sigma_\LL=\cov(\xu_1)$.
  Clearly, $\tr(\Sigma_\LL)=\sigma_\LL^2$.
  An application of the triangle inequality yields
  \[
  |\hat \mu(t;\xi)| \leq \sum_{n \geq 0} \tau_{n} \big(A_{1,n}+A_{2,n}+A_{3,n}\big)
  \]
  with
  \begin{align*}
    A_{1,n}(\xi) &:= \E\Big [\Big|\prod_{j=1}^{n+1}\phi(\alpha_j \beta_{jn}^T\xi) -  \prod_{j=1}^n \gamma(\alpha_j \sqrt{\Sigma_\LL}\beta_{jn}^T\xi) \Big|\Big], \\
    A_{2,n}(\xi) &:= \E\Big [\Big| \prod_{j=1}^{n+1} \gamma(\alpha_j \sqrt{\Sigma_\LL}\beta_{jn}^T\xi) -\prod_{j=1}^{n}  \gamma(\sqrt{a\Sigma_\LL} \beta_{j {n}}^T\xi) \Big| \Big], \\
    A_{3,n}(\xi) &:= \prod_{j=1}^{n+1}  \gamma(\sqrt{a\Sigma_\LL} \beta_{j {n}}^T\xi).
  \end{align*}
  Since the last term is real and positive for all $\xi\in\setR^d$, 
  no absolute value needed.

  \paragraph{\underline{Estimate for $A_{1,n}$}} 
  Another application of the triangle inequality and of \eqref{fact1} 
  gives
  \begin{align*}
    &A_{1,n}(\xi) \leq \sum_{j=1}^{n+1} \E \big|\phi(\alpha_j \beta_{jn}^T\xi)-\gamma(\alpha_j \sqrt{\Sigma_\LL }\beta_{jn}^T\xi) \big|  \\
    & \leq \sum_{j=1}^{n+1} \E \big|\phi(\alpha_j \beta_{jn}^T\xi) - \big(1- \frac{1}{2} \xi^T \alpha_j \beta_{jn} \Sigma_\LL\beta_j^T \xi \big) \big| 
    +  \sum_{j=1}^{n+1} \E \big|\gamma(\alpha_j \sqrt{\Sigma_\LL } \beta_{jn}^T\xi) - \big(1- \frac{1}{2} \xi^T \alpha_j \beta_{jn} \Sigma_\LL\beta_j^T \xi \big) \big| .
  \end{align*}
  The differences inside the absolute values are now estimated 
  using \eqref{eq.taylor2} from the appendix.
  Since $\gamma$ is the characteristic function of a 
  Gaussian random vector $Y$ (with zero mean und unitary covariance matrix),
  this leads to
  \begin{align*}
    A_{1,n}(\xi) &\leq C_p \sum_{j=1}^{n+1}  \left\{ \E\big[ | (\beta_{jn}^T\xi)^T (\xu_j-m)|^p ]
    + \E\big[ | (\beta_{jn}^T\xi)^T  \sqrt{\Sigma_\LL } Y |^p ] \right\} \\
    & \leq C_p |\xi|^p \big(2^p \LL^p+ \|\Sigma_\LL\|_\infty^{p/2}\E[|Y|^p]\big)\sum_j \sup_{|\ee|=1}   \E|\beta_{jn}^T\ee|^p  .
  \end{align*}
  Recalling estimate \eqref{eq.mest} on the sums of $\beta$'s,
  we arrive at 
  \[
  A_{1,n} \Big (\frac{\xi}{r}\Big) \leq C_1 ( \sigma_\LL^p +|\LL|^p)\Gamma_n\Big |\frac{\xi}{r}\Big|^p,
  \quad \text{where} \quad
  \Gamma_n:=\frac{\Gamma(n+\kappa_p)}{n!\Gamma(\kappa_p)}.
  \] 
  \paragraph{\underline{Estimate for $A_{2,n}$}} 
  Since $|e^{-x^2}-e^{-y^2}|\leq |x^2-y^2|$ and since $p/2>1$,
  we have
  \begin{align*}
    |A_{2,n}(\xi)|^{\frac{p}{2}} 
    & \leq \E \Big[ |  e^{ -\frac{1}{2}\sum_{j=1}^n \alpha_j   \xi^T \beta_{jn} \Sigma_\LL \beta_{jn}^T  \xi }
    -e^{ -\frac{1}{2}  \sum_{j=1}^n a \xi^T \beta_{jn} \Sigma_\LL \beta_{jn}^T  \xi }  |^{\frac{p}{2}}   \Big] \\
    &\leq \E\left[\Big|\sum_j \frac{1}{2} (\alpha_j - a)   \xi^T \beta_{jn} \Sigma_\LL \beta_{jn}^T  \xi\Big |^{\frac{p}{2}} \right]
  \end{align*}
  An application of the von Bahr-Esseen inequality \eqref{vonBahrEsseen} reveals
  that
 \begin{align*}
    |A_{2,n}(\xi)|^{\frac{p}{2}}  
    \leq  \frac{2}{2^{\frac{p}{2}} } \E |\alpha_1 - a|^{\frac{p}{2}} 
    \sum_j\E[|\xi^T \beta_{jn} \Sigma_\LL \beta_{jn}^T  \xi |^{\frac{p}{2}} ]
    \leq \frac{2}{2^{\frac{p}{2}} }\E |\alpha_1 - a|^{\frac{p}{2}}  \|\Sigma_\LL\|_\infty^{\frac{p}{2}}  |\xi|^{p}  
    \sum_{j=1}  \sup_{|\ee|=1}  \E|\beta_j^T\ee|^{p}
  \end{align*}
 and finally yields 
  \[
  A_{2,n}\Big(\frac{\xi}{r}\Big) \leq C \Big|\frac{\xi}{r}\Big|^2\sigma_\LL^2    \big ( \Gamma_{n}(\kappa_p)  \big )^{2/p}.
  \]
  \paragraph{\underline{Estimate for $A_{3,n}$}}
  Define $\mu^*(t)$ for $t\geq0$ by
  \[
  \hat \mu^*(t)(\xi):=\sum_{n \geq 0} \tau_{n}
  \E\Big[\prod_{j=1}^{n+1}  \gamma\Big(\sqrt{\frac{d \Sigma_\LL}{\sigma_\LL^2}} \beta_{j {n}}^T  \xi \Big)\Big],
  \]
  which is the solution of the Boltzmann equation with initial datum 
  \[
  \hat\mu^*_0(\xi) = \gamma\Big(\sqrt{\frac{d \Sigma_\LL}{\sigma_\LL^2}}\xi\Big).
  \]
  In particular, each $\mu^*(t)$ has mean zero and unit temperature. 
  Thus, we can apply the quantitative estimate \eqref{eq.fourierconverge}
  for convergence of $\mu^*(t)$ to $\mu_\infty$, 
  obtaining
  \begin{align*}
    \sum_{n \geq 0} \tau_n \E\Big[\prod_{j=0}^{n+1}\gamma(\sqrt{a\Sigma_\LL} \beta_{j {n}}^T \xi/r) \Big]
    & \leq\bigg| \sum_{n \geq 0} \tau_n \E\Big[\prod_{j=1}^{n+1}   
    \gamma\Big(\sqrt{\frac{d \Sigma_\LL}{\sigma_\LL^2}} \beta_{j {n}}^T \frac{\sqrt{a \sigma_\LL^2}}{r\sqrt{d}}\xi \Big)\Big] -
    \Psi\Big ( \frac{\sqrt{a \sigma_\LL^2}}{r\sqrt{d}}\xi \Big)\bigg|+
    \Psi\Big ( \frac{\sqrt{a \sigma_\LL^2}}{r\sqrt{d}}\xi\Big)   \\
    &\leq C \newd_\weight(\mu^*(t),\mu_\infty)\max\Big ( \Big|\frac{\sqrt{a \sigma_\LL^2}}{r\sqrt{d}}\xi\Big |^2,
    \Big| \frac{\sqrt{a \sigma_\LL^2}}{r\sqrt{d}}\xi\Big |^p \Big ) + \Psi\Big ( \frac{\sqrt{a \sigma_\LL^2}}{r\sqrt{d}}\xi\Big) \\
    &\leq C' e^{-(1-\Lambda)t}\max \big ( \frac{\sigma_\LL^2}{r^2},\frac{\sigma_\LL^p}{r^p} \big)\max (|\xi|^p,|\xi|^2)
    + \Psi\Big ( \frac{\sqrt{a \sigma_\LL^2}}{r\sqrt{d}}\xi\Big)
  \end{align*}
  Combining the estimates on $A_{1,n}$ to $A_{3,n}$
  yields the following bound on $I(t)$ defined in \eqref{eq.defIt}:
  \[
  I(t) \leq C\Big ( \sum_{n\geq 0} \tau_n \Big( 
  \big(|\frac{\sigma_\LL}{\RR}|^p +|\frac{\LL}{\RR}|^p \big)\Gamma_n    
  + |\frac{\sigma_\LL}{\RR}|^2 \Gamma_n^{2/p}  \Big)
  + C'e^{-(1-\Lambda)t}\max \big ( \frac{\sigma_\LL^2}{r^2},\frac{\sigma_\LL^p}{r^p} \big)  +
  H\big(\sqrt{\frac{a}{d}} \frac{\sigma_\LL}{\RR}\big ) 
  \Big ).
  \]
  Using the fact that
  \begin{align*}
    \sum_{n=0}^\infty \tau_n \Gamma_{n} = e^{-(1-\kappa_p)t},
  \end{align*}
  and that ---
  since $\tau_n\geq0$ with $\sum_{n=0}^\infty\tau_n=1$, 
  and since $z\mapsto z^{2/p}$ is concave ---
  \begin{align*}
    \sum_{n=0}^\infty \tau_n \Gamma_{n}^{2/p} \leq \bigg( \sum_{n=0}^\infty \tau_n \Gamma_{n} \bigg)^{2/p} = \exp\Big(-\frac2p(1-\kappa_p)t\Big),
  \end{align*}
  we conclude that 
  \begin{align*}
    \sup_{z\in\setR^d}\mu(t)[\ball_\RR(z)]
    &\leq  C  \Big\{ \big ((\sigma_\LL/\RR)^p+(\LL/\RR)^p\big) e^{-(1-\kappa_p)t}  \\
    & \qquad +  (\sigma_\LL/\RR)^2e^{-\frac{2}{p}(1-\kappa_p) t} +(\sigma_\LL/\RR)^2 e^{-(1-\Lambda) t} 
    +H(\sqrt{a}\sigma_\LL/\RR\sqrt{d})\Big\} .
  \end{align*}
  Since $\sigma_\LL^2=\E[(\xu-m)^2]\leq 2 \LL^2$,
  and under the assumption that $\sigma_\LL>\RR\sqrt{2}$,
  we eventually arrive at formula \eqref{eq.explosion}.
\end{proof}

\begin{proof}[Proof of Theorem \ref{thm.explosion}]
  It remains to be proven that estimate \eqref{eq.explosion} implies Theorem \ref{thm.explosion}.
  To this end, observe that for every $\epsilon>0$, 
  there is $\bar \LL$ such that for every $\LL \geq \bar \LL$
  \[
  \sigma_\LL^2 \geq \frac{1-\epsilon}{2} \int_{\ball_\LL(0)} |v|^2 \mu_0(\dd v).
  \]
  See, e.g., \cite{CaGaReExplosion}. 
  Hence, if $\tmp[\mu_0]=+\infty$, then also $\sqrt{a_\LL}\sigma_\LL \to +\infty$ as $\LL \to +\infty$.
  Consequently, the second term in the sum \eqref{eq.explosion} can be made arbitrarily small
  by choosing $\LL$ large enough.
  And the first term in the sum converges, for fixed $\LL$, to zero as $t\to\infty$.
\end{proof}

\section{Examples}
\label{sct.examples}

\subsection{Example: Maxwell molecules}
\label{sct.maxwell}
In this section, we discuss one example for \eqref{eq.b}, 
which provides a model for a particular case of inelastically colliding Maxwellian molecules \cite{Bobylev};
see \cite{CarTos} for a recent review.
The linear collision rules are given by
\begin{align}
  \label{eq.ourrules}
  v' = v + \alpha\nml\cdot(v_*-v)\nml, \quad
  v_*' = v_* - \alpha\nml\cdot(v_*-v)\nml.
\end{align}
These collision rules are brought in the form \eqref{eq.rules}
by setting
\begin{align}
  \label{eq.maxrules}
  L=\alpha\nml\nml^T \quad\text{and}\quad R=\eins-\alpha\nml\nml^T.
\end{align}
Physically, the unit vector $\nml$ points along the contact line of the molecules upon collision.
In accordance with our basic model assumptions, 
$\nml$ is supposed to have a fixed distribution on the sphere $\sphere^{d-1}$
that does not change with $v$ and $v_*$.
For simplicity, we shall consider a uniform distribution, 
but situations with spatial anisotropies could be discussed along the similar lines.
\begin{rmk}
  The rules \eqref{eq.ourrules} are given in the so-called $\omega$-representation,
  which can easily be rephrased in terms of the (standard) $\sigma$-representation,
  see e.g. \cite{Villani},
  \begin{align}
    \label{eq.theirrules}
    v' = \frac{v+v_*}{2} +  \frac{|v-v_*|}{2}\sigma, \quad
    v_*' = \frac{v+v_*}{2} -  \frac{|v-v_*|}{2}\sigma.
  \end{align}
  For Maxwell molecules, one typically assumes an easy distribution for $\sigma$ on $\sphere^{d-1}$,
  referred to as \emph{cross section}, which leads, after the corresponding change of variables, 
  to a complicated distribution for $\nml$.
  In turn, the uniform distribution for $\nml$ assumed here gives rise to the cross section
  \begin{align}
    \label{eq.spheredensity}
    B(\zz,\sigma) = \frac{2^{1-d/2}}{|\sphere^{d-1}|} (1-\zz^T\sigma)^{1-{d/2}}, 
    \quad \text{with} \quad \zz:=\frac{v-v_*}{|v-v_*|}.
  \end{align}
\end{rmk}
The non-negative random variable $\alpha$ in \eqref{eq.ourrules} models the inelasticity of the collision.
Our model is thus similar --- but not identical --- to the one for inelastic Maxwell molecules 
with a background heat bath proposed in \cite{CorCarTos}.
For $\alpha\equiv1$, one obtains a particular model for elastic Maxwell molecules,
with the untypical cross section \eqref{eq.spheredensity}.
The resulting rules (still for $\alpha\equiv1$) are very strict in the sense that they conserve
the particle momenta and kinetic energies in each individual collision,
as becomes evident from
\begin{align}
  \label{eq.strictly}
  L^T L + R^TR = \eins \ a.s.
\end{align}
Indeed, if one chooses the strict correlation $(L,R)=(L_*,R_*)$ in \eqref{eq.rules},
then
\begin{align*}
  v_*+v_*' = (L+R)v + (L+R)v_* = v+v_*,
\end{align*}
and, since $L^TL=L$, $R^TR=R$ and $L^TR=R^TL=0$,
\begin{align*}
  |v'|^2 + |v_*'|^2 
  & = v^TL^TLv + v_*^TR^TRv_* + v^TL^TRv_* + v_*^TL^TRv + |v_*'|^2 \\
  & = v^T(L+R)v + v_*^T(L+R)v_* = |v|^2 + |v_*|^2 .  
\end{align*}
The stationary state $\mu_\infty$ is given by 
the Maxwell distribution $\dd\mu_\infty/\dd v=(2\pi)^{-d/2}\exp(-\frac12 |v|^2)$.
This follows, for instance, since \eqref{eq.strictly} implies that $S=\eins$ a.s. is a solution to \eqref{eq.aux}.
%
%
In the general ``inelastic'' case, where $\alpha$ is not identically one, 
we require in addition that $\alpha$ is a positive random variable 
that is independent of $\nml$ and satisfies
\begin{align}
  \label{eq.alphaone}
  \E[\alpha(1-\alpha)]=0.
\end{align}
For simplicity, we also assume that $\alpha$ is ``concentrated near one''
in the sense that
\begin{align}
  \label{eq.alphac}
  \E[\alpha^2(1-\alpha)^2] \leq c\E[\alpha^2] 
\end{align}
for a constant $c>0$ determined below.

In preparation of the following calculations, 
introduce the random variables $V_1$ to $V_d$ by
\begin{align}
  \label{eq.defineV}
  V_i = \ee_i^T\nml,
\end{align}
where $\nml$ is uniformly distributed on the sphere $\sphere^{d-1}\subset\setR^d$.
Notice that $V_1^2 + \cdots + V_d^2 = 1$ almost surely since $\nml$ is a unit vector.
Thus, $1 = \E[V_1^2+\cdots+V_d^2] = d\E[V_1^2]$ since the $V_i$ are identically distributed,
and so $\E[V_1^2] = 1/d$.
For later reference, we also remark that
\begin{align}
  \label{eq.v4}
  1 = \E[(V_1^2+\cdots+V_d^2)^2] = d\E[V_1^4] + d(d-1)\E[V_1^2V_2^2]
\end{align}
since each product $V_iV_j$ with $i\neq j$ has the same distribution as $V_1V_2$,
from where it follows that $\E[V_1^4]<1/d$.

We are now going to verify Assumptions 1--3 for the model at hand.
For notational simplicity,
we shall directly work with the unsymmetrized matrices \eqref{eq.maxrules}.

\subsubsection{Assumption \ref{ass.1}}
Since $\nml\nml^T$ defines an orthogonal projector for every $\nml\in\sphere^{d-1}$,
and recalling that $\nml$ and $\alpha$ are independent,
it follows that
\begin{align*}
  \E[L^TL+R^TR] = \eins - 2\E[\alpha(1-\alpha)]\E[\nml\nml^T].
\end{align*}
Condition \eqref{eq.alphaone} yields the desired result.
%

\subsubsection{Assumption \ref{ass.3new}}
We are going to prove \eqref{eq.weirdass} for the weight function $\weight(\xi)=|\xi|^4$.
For an arbitrary unit vector $\ee$, we have
\begin{align*}
  \E[|L^T\ee|^4+|R^T\ee|^4] 
  &= \E[\alpha^4]\E[V_1^4] + 1 - 2\E[\alpha(2-\alpha)]\E[V_1^2] + \E[\alpha^2(2-\alpha)^2]\E[V_1^4] \\
  &= 1 - 2\E[\alpha^2]\E[V_1^2] + 2\E[\alpha^4-2\alpha^3+2\alpha^2]\E[V_1^4] \\
  &= 1 - 2\E[\alpha^2](\E[V_1^2]-\E[V_1^4]) + 2\E[\alpha^2(1-\alpha)^2]\E[V_1^4] .
\end{align*}
Assuming that the constant $c$ in \eqref{eq.alphac} satisfies $c<\E[V_1^4]\leq\E[V_1^2]-\E[V_1^4]$,
inequality \eqref{eq.weirdass} follows.


\subsubsection{Assumption \ref{ass.2}}
We already know that $\Sigma_*=\eins$.
Let a symmetric matrix $M\in\smat$ be given, with $\tr M=0$.
Since the distribution of $(L,R)$ is rotationally symmetric,
we may assume that $M$ is diagonal with respect to the standard basis,
\begin{align}
  \label{eq.decomposeM}
  M = \sum_{k=1}^d \mu_k \ee_k\ee_k^T, \qquad \text{with} \quad \sum_{k=1}^d \mu_k = 0. 
\end{align}
Now, let some unit vector $\ee\in\sphere^{d-1}$ be given, and introduce the numbers $\lambda_i=\ee_i^T\ee$,
satisfying $\lambda_1^2+\cdots+\lambda_d^2=1$.
We obtain, recalling that $\alpha$ and $\nml$ are independent,
\begin{align*}
  \ee^T\E[L M L^T]\ee = \sum_{i,j,k=1}^d \lambda_i\lambda_j\mu_k\E[\ee_i^TL\ee_k\ee_k^TL^T\ee_j] 
  = \E[\alpha^2] \sum_{i,j,k=1}^d \lambda_i\lambda_j\mu_k \E[V_iV_jV_k^2] .
\end{align*}
At this point, observe that
\begin{align*}
  \E[V_iV_jV_k^2] = \delta_{ij} \E[V_1^2V_2^2] + \delta_{ij}\delta_{jk} ( \E[V_1^4]-\E[V_1^2V_2^2] ).
\end{align*}
Thus,
\begin{align*}
  \ee^T\E[LML^T]\ee   
  &= \E[\alpha^2] \sum_{i,k=1}^d \lambda_i^2\mu_k \E[V_1^2V_2^2] + \E[\alpha^2] \sum_{k=1}^d \lambda_k^2\mu_k ( \E[V_1^4]-\E[V_1^2V_2^2] ) \\
  & = \E[\alpha^2]\E[V_1^2V_2^2] \sum_{i=1}^d\lambda_i^2 \underbrace{\sum_{k=1}^d \mu_k}_{=0} + \E[\alpha^2](\E[V_1^4]-\E[V_1^2V_2^2])\sum_{k=1}^d \lambda_k^2\mu_k.
\end{align*}
For the evaluation of the analogous term with $R$ in place of $L$,
observe that
\begin{align*}
  \ee_i^TR\ee_k = \ee_i^T\ee_k - \alpha\ee_i^T\nml \nml^T\ee_k = \delta_{ik} - \alpha V_iV_k.
\end{align*}
Using that $\E[V_iV_j] = \delta_{ij}\E[V_1^2]$ and performing similar manipulations as above
yields
\begin{align*}
  \ee^T\E[RMR^T]\ee &= \sum_{i,j,k=1}^d \lambda_i\lambda_j\mu_k\E[\ee_i^TR\ee_k\ee_k^TR^T\ee_j] \\
  & = \sum_{i,j,k=1}^d \lambda_i\lambda_j\mu_k \big( \delta_{ik}\delta_{jk} - \alpha\delta_{ik}\E[V_jV_k] - \alpha\delta_{jk} \E[V_iV_k] + \alpha^2\E[V_iV_jV_k^2] \big) \\
  & = \sum_{k=1}^d \lambda_k^2\mu_k - 2\E[\alpha]\sum_{i,k=1}^d \lambda_i\lambda_k\mu_k\E[V_iV_k] + \E[\alpha^2]\sum_{i,j,k=1}^d \lambda_i\lambda_j\mu_k\E[V_iV_jV_k^2] \\
  & = \big\{1-2\E[\alpha]\E[V_1^2]+\E[\alpha^2]\big(\E[V_1^4]-\E[V_1^2V_2^2]\big)\big\} \sum_{k=1}^d \lambda_k^2\mu_k .
\end{align*}
Adding up, and using that $\E[\alpha^2]=\E[\alpha]$ by \eqref{eq.alphaone}, we have
\begin{align*}
  \ee^T\E[LML^T+RMR^T]\ee 
  = \big\{1-2\E[\alpha]\big(\E[V_1^2]+\E[V_1^2V_2^2]-\E[V_1^4]\big)\big\} \sum_{k=1}^d \lambda_k^2\mu_k
  \leq \kappa \max_{1\leq k\leq d} |\mu_k|,
\end{align*}
with the constant
\begin{align*}
  \kappa = 1-2\E[\alpha]\big(\E[V_1^2]+\E[V_1^2V_2^2]-\E[V_1^4]\big).
\end{align*}
Here we have used implicitly that $\kappa\geq0$, 
which is true since
\begin{align*}
  \E[V_1^2V_2^2] \leq \E[V_1^4] \quad \text{and} \quad \E[V_1^2] = 1/d \leq 1/2,
\end{align*}
and since $\E[\alpha]\leq1$, which follows from $\E[\alpha]^2\leq\E[\alpha^2]=\E[\alpha]$.
To conclude Assumption \ref{ass.2}, it suffices to check that indeed $\kappa<1$.
But this is a trivial consequence of $\E[V_1^2]=1/d$ and $\E[V_1^4]<1/d$,
which in turn follows from \eqref{eq.v4}.

\subsubsection{Heavy tails}
Above, we have shown that the fourth absolute moment of $\mu_\infty$ is finite if $c$ in \eqref{eq.alphac} is small enough.
Next, we argue that unless $\alpha\equiv1$ a.s., there exists a $s>4$ such that the $s$th absolute moment of $\mu_\infty$ diverges.
In other words, the steady distribution $\mu_\infty$ develops fat tails as soon as the collision rules are ``a bit inelastic''.

The goal is to show \eqref{eq.bbstar} for some (sufficiently large) $s>4$.
Observe that, to this end, it suffices to prove
\begin{align}
  \label{eq.bbb}
  \E[|R\ee_1|^s] = \E[\alpha^s]\E[V_1^s] > 1.
\end{align}
Since we assume that \emph{not} $\alpha\equiv1$ a.s.,
it follows from \eqref{eq.alphaone} that $\alpha>1$ with positive probability.
Consequently, there exists some $\epsilon^*\in(0,1/2)$ such that $p^*:=\prb(\alpha>1+2\epsilon^*)>0$.
Observe further that $q^*:=\prb(V_1>1-\epsilon^*)>0$,
and so
\begin{align*}
  \E[\alpha^s]\E[V_1^s] \geq p^*q^*(1+2\epsilon^*)^s(1-\epsilon^*)^s \geq p^*q^*(1+\epsilon^*/2)^s,
\end{align*}
tends to infinity as $s\to\infty$.
So \eqref{eq.bbb} holds for some large $s$.

%

\subsubsection{Regularity}

We have seen that outside of the elastic case, 
the stationary distribution $\mu_\infty$ is \emph{not} a Gaussian, but a scale mixture.
Depending on the distribution of $\alpha$, this mixture might or might not possess a smooth density.

In order to check the applicability of Theorem \ref{thm.smooth}, 
we need to evaluate the criteria \eqref{eq.smoothcon1} and \eqref{eq.smoothcon2}.
By rotational symmetry, it suffices to calculate the expectations for \emph{one particular} unit vector, say $\ee=\ee_1$.
For \eqref{eq.smoothcon1}, we find
\begin{align*}
  \E\big[ \min(|L^T\ee_1|,|R^T\ee_1|)^{-\delta}\big]
  &\le \E\big[\alpha^{-\delta}\big] \big(\E\big[|V_1|^{-\delta}\big] + \E\big[(1-V_1^2)^{-\delta/2}\big]\big) \\
  &= \sigma_d\E[\alpha^{-\delta}]\bigg(\int_{-1}^1 |z|^{-\delta}(1-z^2)^{(d-3)/2}\dd z + \int_{-1}^1 (1-z^2)^{(d-3-\delta)/2}\dd z \bigg).
\end{align*}
Here $\sigma_d$ is the ration of the respective hypersurface measures of $\sphere^{d-1}$ and $\sphere^{d-2}$.
The sum of the last integrals is obviously finite for as long as $\delta<1$.
Thus \eqref{eq.smoothcon1} holds if $\E[\alpha^{-\delta}]$ is finite for some $\delta>0$.

On the other hand, since $|L^T\ee_1|^2+|R^T\ee_1|^2=\alpha^2$ a.s.,
it follows that
\begin{align*}
  \E\big[ \max(|L^T\ee_1|,|R^T\ee_1|)^{-\bar a}\big]
  \ge \E\big[ (\alpha^2/2)^{-\bar a/2}\big] = 2^{\bar a/2} \E\big[ \alpha^{-\bar a} \big].
\end{align*}
In conclusion, the regularity of $\mu_\infty$ is entirely determined by the largest exponent $\bar a$ 
for which $\E[\alpha^{-\bar a}]$ is finite.
If, for instance $\alpha\ge\epsilon>0$ a.s.,
then $\bar a>0$ is arbitrary,
and $\mu_\infty$ possesses a density in $C^\infty(\setR^d)$.

\subsection{Random rotations}
Let $A,\,B$ be random matrices in $\so(\setR^d)$, 
i.e., $A$ and $B$ represent random rotations in $\setR^d$.
Further, let $\alpha,\beta$ be non-negative random variables with
\begin{align}
  \label{eq.exampleab}
  \E[\alpha^2+\beta^2]=1 \quad \text{and} \quad \kappa_p:=\E[\alpha^p+\beta^p]<1
\end{align}
for some $p\in(2,3)$.
For simplicity, we assume further that $(\alpha,\beta)$ is independent of $(A,B)$.
Now define
\begin{align*}
  L=\alpha A, \quad R=\beta B.
\end{align*}
This is a relatively straight-forward extension of a random kinetic model on $\setR$ to $\setR^d$. 
As already pointed out in the introduction, 
examples of this type might violate Assumption \ref{ass.2}
if the rotations leave certain genuine subspaces of $\setR^d$ invariant.
We claim that our set assumptions is met if the following is true for every unit vector $\ee$:
the measure $\mu_\ee$ induced on the sphere $\sphere^{d-1}$ through $A^T\ee$,
\begin{align*}
  \mu_\ee(\mathcal{B}) = \prb( A^T\ee\in\mathcal{B} )
  \quad \text{for all Borel sets $\mathcal{B}\subset\sphere^{d-1}$},
\end{align*}
can be decomposed in the form $\mu_\ee=\epsilon\sigma+(1-\epsilon)\tilde\mu_\ee$,
where $\sigma$ is the normalized surface measure on $\sphere^{d-1}$,
and $\tilde\mu_\ee$ is some probability measure,
with an $\epsilon>0$ independent of $\ee$.

Since $A,\,B$ are rotation matrices, $A^TA=B^TB=\eins$,
and so
\begin{align*}
  \E\big[L^TL + R^TR \big] = \E[\alpha^2+\beta^2] = 1,
\end{align*}
which shows Assumption \ref{ass.1}.
Also, with the same $p$ as in \eqref{eq.exampleab},
\begin{align*}
  \E\big[ |L^T\xi|^p + |R^T\xi|^p \big] = \E[\alpha^p+\beta^p]|\xi|^p = \kappa_p|\xi|^p,
\end{align*}
which is Assumption \ref{ass.3new} with $\omega(\xi)=|\xi|^p$.
To check Assumption \ref{ass.2}, 
first note that $\Sigma_*=\eins$ satisfies
\begin{align*}
  \E[L\Sigma_*L^T+R\Sigma_*R^T] = \E[\alpha^2AA^T+\beta^2BB^T] = \Sigma_* .
\end{align*}
Next, let a symmetric matrix $M\in\smat$ be given.
Then, for every $\ee\in\sphere^{d-1}$,
\begin{align*}
  \ee^T\E\big[LML^T+RMR^T\big]\ee
  & = \E\big[\alpha^2(A^T\ee)^T\Sigma(A^T\ee)+\beta^2(B^T\ee)^T\Sigma(B^T\ee)\big] \\
  & = \E[\alpha^2+\beta^2]\int_{\sphere^{d-1}} w^TMw\mu_\ee(\dd w) \\
  & = (1-\epsilon)\int_{\sphere^{d-1}} w^TMw\tilde\mu_\ee(\dd w) + \epsilon \int_{\sphere^{d-1}} w^TM w \sigma(\dd w) ,
\end{align*}
where the assumed decomposition of $\mu_\ee$ has been employed.
The last integral equals to $\tr M/d$, and thus vanishes if $M$ is traceless.
This is easily verified by decomposing $M$ in the form \eqref{eq.decomposeM},
defining $V_1,...,V_d$ as in \eqref{eq.defineV} with $w$ in place of $\nml$, 
and observing that the integral equals to a sum over $\mu_k\E[V_k^2]$ for $k=1,...,d$.
This means that
\begin{align*}
  \big|\ee^T\E\big[LM L^T+RMR^T\big]\ee \big|
  \leq 
  (1-\epsilon) \int_{\sphere^{d-1}} \|M\|_\infty\tilde\mu_\ee(\dd w) 
  = (1-\epsilon)\|M\|_\infty, 
\end{align*}
proving Assumption \ref{ass.2}.
Thus, our theory provides the existence and uniqueness of a stationary distribution $\mu_\infty$.

Once this is known, it is now easy to conclude that $\mu_\infty$ is --- as one would expect ---
the rotationally symmetric extension of the stationary state $\mu_*\in \meas(\setR)$ 
for the corresponding \emph{one-dimensional} model with scalar coefficients $(L',R')=(\alpha,\beta)$.
From the results in \cite{BasLadMat}, it follows that the one-dimensional equation
possesses a unique stationary solution $\mu_*$ that is centered and of unit second moment;
call its characteristic function $\psi$.
We claim that $\Psi$, the characteristic function of $\mu_\infty$, 
satisfies $\Psi(\xi):=\psi(|\xi|)$ for all $\xi\in\setR^d$.
In fact, this follows immediately since
\begin{align*}
  \E[\Psi(L^T\xi)\Psi(R^T\xi)] = \E[\psi(\alpha|A^T\xi|)\psi(\beta|B^T\xi|)]
  = \E[\psi(\alpha|\xi|)\psi(\beta|\xi|)] = \psi(|\xi|) = \Psi(\xi).
\end{align*}

\subsection{An asymmetric example}
\label{sct.cross}
The last example equation we consider produces a steady distribution
that is \emph{not} rotationally symmetric.
Thus, even if rotationally symmetric initial conditions are used,
the results from \cite{BobCerGam} do not apply,
as the symmetry is broken at any $t>0$.
Specifically, we state the example in the plane $\setR^2$,
and the stationary distribution $\mu_\infty$ will be concentrated 
on the coordinate cross $\{(v_1,v_2)\in\setR^2|v_1v_2=0\}$.
(It is not hard to generalize the idea developed below to higher dimensions 
and to more sophisticated non-symmetric stationary states.)

Recall that $\ee_1,\,\ee_2$ are the canonical basis vectors,
and introduce the vectors $\ee_+,\,\ee_-$ by $\ee_\pm=(\ee_1\pm\ee_2)/\sqrt{2}$.
The random coefficients $L$ and $R$ are defined by
\begin{align*}
  (L,R) = \big(\ee_i\ee_+^T,\ee_i\ee_-^T\big),
\end{align*}
where the random index $i$ attains the values 1 and 2 with probabilities $q$ and $1-q$, respectively.
Since $\ee_+$ and $\ee_-$ are orthonormal, 
one immediately verifies that
\begin{align*}
  L^TL + R^TR = \ee_+\ee_i^T\ee_i\ee_+^T + \ee_-\ee_i^T\ee_i\ee_-^T = \ee_+\ee_+^T + \ee_-\ee_-^T = \eins
\end{align*}
almost surely.
Let $\Sigma\in\smat$ be given, with
\begin{align}
  \label{eq.typicalsigma}
  \Sigma = \begin{pmatrix} \sigma_+ & \tau \\ \tau & \sigma_- \end{pmatrix} .
\end{align}
One finds that
\begin{align*}
  \tq(\Sigma) = \E\big[L\Sigma L^T+R\Sigma R^T\big]
  =\big( \ee_+^T\Sigma\ee_+ + \ee_-^T\Sigma\ee_- \big) \E\big[ \ee_i\ee_i^T \big]
  = \frac{\sigma_++\sigma_-}2 \Sigma_*, 
\end{align*}
with the fixed point $\Sigma_*$ of $\tq$ given by
\begin{align*}
  \Sigma_* = 2 \begin{pmatrix} q & 0 \\ 0 & 1-q \end{pmatrix} .
\end{align*}
In particular,
$\tq(\Sigma)=0$ for traceless matrices $\Sigma$, i.e., for $\sigma_++\sigma_-=0$.
Thus, Assumption \ref{ass.2} holds with $\kappa=0$.
Define a weight function $\omega:\setR^2\to\setR$ by setting $\omega(\lambda\ee_\pm)=\lambda^4$ for all $\lambda\in\setR$,
and $\omega(\xi)=\check\omega|\xi|^4$ for all $\xi\in\setR^2$ that are \emph{not} parallel to $\ee_+$ or $\ee_-$,
with a value $\check\omega\geq1$ yet to be determined.
For arbitrary $\xi=(\xi_1,\xi_2)\in\setR^2$, it follows that
\begin{align*}
  g(\xi) := \E\big[\omega(L^T\xi) + \omega(R^T\xi) \big]
  & = q \big( \omega(\xi_1\ee_+) + \omega(\xi_1\ee_-) \big) + (1-q) \big( \omega(\xi_2\ee_+) + \omega(\xi_2\ee_-) \big) \\
  & = 2( q\xi_1^4+ (1-q)\xi_2^4).
\end{align*}
Now if $\xi=\lambda\ee_\pm$, then $\xi_1^4=\xi_2^4=\lambda^4/4$, 
and so
\begin{align*}
  g(\xi) = \frac12 \lambda^4 = \frac12 \omega(\xi).
\end{align*}
If instead $\xi$ is not parallel to $\ee_+$ or $\ee_-$, 
then we estimate
\begin{align*}
  g(\xi) = 2\frac{q\xi_1^4+(1-q)\xi_2^4}{(\xi_1^2+\xi_2^2)^2}|\xi|^4\leq \frac{2\max(q,1-q)}{\check\omega}\omega(\xi).
\end{align*}
Choosing $\check\omega:=4\max(q,1-q)$, 
it then follows that
\begin{align*}
  \E\big[\omega(L^T\xi) + \omega(R^T\xi) \big] \leq \frac12\omega(\xi)
\end{align*}
for \emph{all} $\xi\in\setR^2$,
verifying Assumption \ref{ass.3new} with $p=4$ and $\kappa_p=1/2$.
\begin{rmk}
  We emphasize that \eqref{eq.weirdass} would \emph{not} hold
  with some $\kappa_p<1$ if we had made the canonical choice $\omega(\xi)=|\xi|^4$.
\end{rmk}
Consequently, the equation \eqref{eq.b} possesses 
a unique stationary distribution $\mu_\infty$ of unit temperature:
$\mu_\infty$ is supported on the coordinate cross
and satisfies
\begin{align}
  \label{eq.crossstate}
  \int_{\setR^2}\varphi(v)\dd\mu_\infty(v) = (4\pi)^{-1/2}\int_\setR \big( q \varphi(w,0) + (1-q) \varphi(0,w) \big)e^{-w^2/4}\dd w 
\end{align}
for all test functions $\varphi\in C^0_b(\setR^2)$.
The fact that \eqref{eq.crossstate} defines an equilibrium 
is most easily checked on the level of its Fourier transform,
\begin{align*}
  \Psi(\eta) := \hat\mu_\infty(\eta) = q e^{-\eta_1^2} + (1-q) e^{-\eta_2^2} .
\end{align*}
Indeed, $\Psi$ satisfies the stationary equation \eqref{eq.stationary} for every $\xi\in\setR^2$:
\begin{align*}
  \E\big[\Psi(L^T\xi)\Psi(R^T\xi)\big]
  & = q \Psi(\xi_1\ee_+)\Psi(\xi_1\ee_-) + (1-q) \Psi(\xi_2\ee_+)\Psi(\xi_2\ee_-)  \\
  & = q \big( q e^{-\xi_1^2/2} + (1-q) e^{-\xi_1^2/2} \big)^2 
  + (1-q) \big( q e^{-\xi_2^2/2} + (1-q) e^{-\xi_2^2/2} \big)^2 = \Psi(\xi).
\end{align*}

\section{Appendix}
Here we list a variety of inequalities that are used throughout our calculations.

\begin{itemize}
\item \textbf{Young's inequality:} for every $p>2$, and for all $x,y\ge0$,
  \begin{align}
    \label{eq.young}
    x^{p/2-1}y \leq \frac{p-2}{p} x^{p/2} + \frac{2}{p} y^{p/2} .
  \end{align}
\item \textbf{Inequality on moments:} for every $s>1$,  there is a constant $C_s$ 
  such that for all $x,y\ge0$,
  \begin{align}
    \label{eq.myineq}
    (x+y)^s \leq x^s+y^s+C_s(x^{s-1}y+xy^{s-1}).
  \end{align}
\item \textbf{Remainder estimate for Fourier transforms:} for $p\in(2,3)$, and for all $u\in\setR$,
  \begin{align}
    \label{eq.taylor1}
    \big|e^{-iu}-(1-iu)\big| \leq \min\{2|u|, \half|u|^2 \}\leq 2^{3-p}|u|^{p/2}, \\
    \label{eq.taylor2}
    \big|e^{-iu}-(1-iu+\frac{1}{2} u^2)\big| \leq \min\{u^2, {\textstyle\frac16}|u|^3  \} \leq 6^{2-p}|u|^p.
  \end{align}
\item \textbf{Von Bahr and Esseen inequality:}
  Let $Z_1,\dots,Z_n$ independent (real valued) random variables such that $\E[Z_i]=0$ and
  $\E[|Z_i|^p]<+\infty$ for some $1 \leq p \leq 2$, then
  \begin{equation}\label{vonBahrEsseen}
    \E[|\sum_{i=1}^n Z_i|^p] \leq 2\sum_{i=1}^n \E[|Z_i|^p].
  \end{equation}
  See \cite{BahrEsseen1965} for a proof.
\end{itemize}

\end{document}